\documentclass[11pt,journal,draftcls,onecolumn,peerreviewca]{IEEEtran}
\usepackage{amsfonts}
\usepackage{amsmath}
\usepackage{amssymb}
\usepackage{bm}
\usepackage{xcolor,graphicx,float}
\usepackage{color, soul}
\usepackage{stfloats}
\usepackage[numbers,sort&compress]{natbib}
\usepackage[amsmath,thmmarks]{ntheorem}
\usepackage{theorem}
\usepackage{algorithm}
\usepackage{algorithmic}
\usepackage{graphicx}
\usepackage{subfigure}
\usepackage{pict2e}
\usepackage{mathrsfs}
\newtheorem{theorem}{Theorem}
\newtheorem{lemma}{Lemma}
\newtheorem{remark}{Remark}

\theoremheaderfont{\sc}\theorembodyfont{\upshape}
\theoremstyle{nonumberplain}
\theoremseparator{}
\theoremsymbol{\rule{1ex}{1ex}}
\newtheorem{proof}{Proof}

\begin{document}
\title{Multi-snapshot Newtonized Orthogonal Matching Pursuit for Line Spectrum Estimation with Multiple Measurement Vectors}

\author{Jiang~Zhu Lin Han, Rick S. Blum and Zhiwei Xu \thanks{Jiang Zhu, Lin Han and Zhiwei Xu are with the Key Laboratory
of Ocean Observation-imaging Testbed of Zhejiang Province, Ocean College,
Zhejiang University, No.1 Zheda Road, Zhoushan, 316021, China. Rick S. Blum is with Electrical and Computer Engineering, Lehigh University, USA.}}
\maketitle
\begin{abstract}
In this paper, multi-snapshot Newtonized orthogonal matching pursuit (MNOMP) algorithm is proposed to deal with the line spectrum estimation with multiple measurement vectors (MMVs). MNOMP has the low computation complexity and state-of-the-art performance advantage of NOMP, and also includes two key steps: Detecting a new sinusoid on an oversampled discrete Fourier transform (DFT) grid and refining the parameters of already detected sinusoids to avoid the problem of basis mismatch. We provide a stopping criterion based on the overestimating probability of the model order. In addition, the convergence of the proposed algorithm is also proved. Finally, numerical results are conducted to show that the performance of MNOMP benefits from MMVs, and the effectiveness of MNOMP when compared against the state-of-the-art algorithms in terms of frequency estimation accuracy and computation complexity.
\end{abstract}
{\bf keywords}: Orthogonal matching pursuit, frequency estimation, line spectrum, Newton refinement, multiple measurement vectors
\section{Introduction}
\label{sec:intro}
One of the classical problems in digital communication and radar processing applications is to estimate continuous-valued frequencies of sinusoids in additive white Gaussian noise (AWGN) environments from a small number of measurements \cite{poor1, Kay1}. On the one hand, several classical subspace methods have been proposed to perform the frequency estimation, such as MUSIC and ESPRIT \cite{Schmidt, Roy}, which exploit the autocorrelation matrix's low-rank structure to estimate the underlying signal subspace. As the signal-to-noise ratio (SNR) decreases, the performance of these two algorithms for estimating closely spaced frequencies will degrade \cite{Teutsch}. On the other hand, a variety of methods based on sparse representation and compressed sensing (CS) have also been proposed to estimate frequencies for multiple sinusoids \cite{Malioutov, Hyder}. Basically, the estimation problem can be transformed to that of seeking a sparse approximation of the multiple sinusoids by referring to an infinite-dimensional dictionary. In fact, of all the frequencies lying on the discrete Fourier transform (DFT) grid, it can be shown that the signal can exactly be recovered by utilising convex optimization from randomly selected samples with high probability \cite{Candes1}.

However, there exists a major grid mismatch problem induced by the fact that the measurements are sparsely represented under a finite discrete dictionary, which badly deteriorates the performance of various reconstruction algorithms. In fact, one has to make a reasonable tradeoff between the oversampling rate and the computational cost when implementing sparse methods. This unavoidable grid mismatch problem is studied in \cite{Chi1, Scharf1} in detail. Moreover, sparse reconstruction methods usually entail one or more parameters which in fact are not necessarily known, such as the number of the sinusoids, the regularization parameters, the variance of the noise and so on. Recently, the semiparametric iterative covariance-based estimation (SPICE) algorithm \cite{Stoica1, Stoica2, Stoica3} has been proposed to alleviate the drawbacks of the discretization operation to a great extent, which uses the covariance fitting criterion from a statistical perspective and no user-parameters are required.


\subsection{Related work}
Recent works have shown that performance can be improved with multiple measurement vectors (MMVs) by harnessing group sparsity \cite{Atom1, Atom2, Atom3,Yang1, Chi, Fang}. In \cite{Yang1}, a sparse and parametric approach (SPA) for uniform and sparse linear arrays is proposed, which utilises MMVs to perform line spectrum estimation by solving a semidefinite programming problem. In \cite{Chi}, two approaches are developed to solve the problem of line spectrum denoising and estimation, which estimate an ensemble of spectrally-sparse signals composed of the same set of continuous-valued frequencies from MMVs, and demonstrate the benefit of including MMVs. The iterative reweighed approach (IRA) is proposed to deal with both single measurement vector (SMV) and MMVs \cite{Fang}, where all the frequencies are updated in parallel.

In \cite{Mamandipoor}, a fast sequential Newtonized orthogonal matching pursuit (NOMP) algorithm is proposed. It is shown that NOMP achieves high estimation accuracy for millimeter wave spatial channel estimation \cite{Madhow1, Madhow2, Madhow3, Jinshi}. Motivated by NOMP and its related applications \cite{Mamandipoor}, we develop the multi-snapshot NOMP (MNOMP) for line spectrum estimation with MMVs, and apply MNOMP for DOA estimation.
\subsection{Contributions}
The main contributions are summarized as follows:
\begin{itemize}
  \item We develop MNOMP algorithm to deal with line spectrum estimation with MMVs. Similar to NOMP \cite{Mamandipoor}, our proposed algorithm avoids the basis mismatch problem by using a Newton refinement step as feedback to improve the estimation of already detected sinusoids. A stopping criterion based on overestimating probability is provided and the model order $K$ is determined. In addition, probability of miss is also calculated. It is numerically shown that MNOMP benefits from MMVs.
  \item We analyze the convergence of MNOMP by using the theory of dual norms. Specifically, we provide an upper bound on the number of iterations required by MNOMP and also obtain a bound on the convergence rate.
  \item Numerical simulations are conducted to show that MNOMP benefits from MMVs. By comparing the estimation performance of various algorithms against the Cram\'er-Rao bound (CRB), MNOMP achieve a near-optimal performance in terms of estimation accuracy. In addition, MNOMP is applied to the DOA problems and its effectiveness is validated.
\end{itemize}

\emph{Outline:}
In Section $\rm {\uppercase\expandafter{\romannumeral2}}$, we set up the problem model. We propose MNOMP in Section $\rm {\uppercase\expandafter{\romannumeral3}}$. In Section $\rm {\uppercase\expandafter{\romannumeral4}}$, we present the stopping criterion based on the probability of overestimating the model order, along with an analytical expression of the miss probability of detecting the sinusoids. We present the  convergence analysis in Section $\rm {\uppercase\expandafter{\romannumeral5}}$. In Section $\rm {\uppercase\expandafter{\romannumeral6}}$, we conduct numerical experiments to compare the estimation accuracy of MNOMP against the state-of-the-art methods. Section $\rm {\uppercase\expandafter{\romannumeral7}}$ concludes the paper.

\emph{Notation:} Let $(\cdot)^{\rm H}$, $(\cdot)^{\rm T}$ and $(\cdot)^{\rm *}$ denote the conjugate transpose, transpose and conjugate operator respectively. $\mathcal {CN}$ denotes the complex Gaussian distribution. The Frobenius norm, the real (imaginary) part of the complex number $a$ and the trace operator are denoted by $\Vert {\cdot} \Vert_{\rm{F}}$, $\Re\{a\}$ ($\Im\{a\}$) and $\rm {tr}\{\cdot\}$ respectively. $\lfloor a \rfloor$ denotes the greatest integer that is less than or equal to $a$. p
\section{Problem Setup}
\label{sec:pro}
In an MMV model, we consider a line spectrum estimation scenario with $N$ sensors and $T$ snapshots {\footnote{{Snapshot is usually used in the array processing environment.}}} collecting measurements of $K$ distinct frequency components. The measurements at the array output can be expressed as \footnote{Extension to the compressed observation scenario ${\mathbf Y}={\boldsymbol \Phi}{\mathbf A}{\mathbf X}+{\mathbf Z}$ is straightforward and is omitted.}
\begin{align}\label{multisnap}
{\mathbf Y}={\mathbf A}{\mathbf X}+{\mathbf Z},
\end{align}
where the $k$th column of $\mathbf A$ is
\begin{align}\label{aw}
{\mathbf a}(\omega_k) \triangleq \frac{1}{\sqrt{N}}\left[1,e^{{\rm j}\omega_k},\cdots,e^{{\rm j}(N-1)\omega_k}\right]^{\rm T},
\end{align}
${\mathbf Y}\in{\mathbb {C}}^{N\times T}$ is the noisy measurement collected by all $N$ sensors and $T$ snapshots, and $\mathbf A=[\mathbf a(\omega_1), \cdots, \mathbf a(\omega_K)]\in {\mathbb {C}}^{N\times K}$. Each $\omega_l$ is continuous-valued in $[0,2\pi)$. ${z}_{ij}$ is independent and identically distributed (i.i.d.) Gaussian random variable and follows ${z}_{ij}\sim {\mathcal {CN}}({0},{\sigma}^2)$. ${\mathbf X}=[\mathbf x_1,\cdots, \mathbf x_T]\in{\mathbb {C}}^{K\times T}$ contains all the sinusoid amplitudes ${x}_{ij}$ for each snapshot. The uniform linear array (ULA) scenario for DOA can also be formulated as model (\ref{multisnap}) \cite{Yang1}.

\section{MNOMP Algorithm}
\label{sec:algorithm}
We first look into the estimation problem of a single sinusoid, and then generalize the results to a mixture of sinusoids. We borrow the idea in NOMP \cite{Mamandipoor} and develop MNOMP for the MMVs.
\subsection{Single frequency}
In this scenario, model (\ref{multisnap}) simplifies to \footnote{For the compressed scenario, the algorithm can be designed via solving ${\mathbf Y}={\boldsymbol \Phi}{\mathbf a}{\mathbf x}^{\rm T}+{\mathbf Z}$.}
\begin{align}\label{two}
{\mathbf Y}={\mathbf a}{\mathbf x}^{\rm T}+{\mathbf Z},
\end{align}
where $\mathbf x=\left[x_1,\cdots,x_T\right]^{\rm T}\in{\mathbb C}^{T\times 1}$, and ${\mathbf a}=\left[1,e^{{\rm j}\omega},\cdots,e^{{\rm j}(N-1)\omega}\right]^{\rm T}/\sqrt{N}$.

The Maximum Likelihood (ML) estimate of the amplitudes $\mathbf x$ and frequency $\omega$ can be calculated by minimizing the residual power $\Vert {\mathbf Y}-{\mathbf a}{\mathbf x}^{\rm T} \Vert^2_{\rm{F}}$, which is equal to maximizing the function
\begin{align}\label{objective}
S({\mathbf x},\omega)=\sum\limits_{t=1}^T2\Re\{{\mathbf y}_t^{\rm H}{x}_t{\mathbf a}\}-|x_t|^2\|{\mathbf a}\|^2\triangleq \sum\limits_{t=1}^TS(x_t,\omega),
\end{align}
which will lead to a simpler description of the algorithm. It's difficult to directly optimize $S({\mathbf x},w)$ over all amplitudes and frequency. As a result, a two stage procedure is adopted: (1) Detection stage, in which a coarse estimate of $\omega$ is found by restricting it to a discrete set, (2) Refinement stage, where we iteratively refine the estimates of amplitudes and frequency of detected sinusoids. Similar to \cite{Mamandipoor}, for any given $\omega$, the gain that maximizes $S(x_t,\omega)$ is $\hat{x}_t=\left({\mathbf a}^{\rm H}{\mathbf y}_t\right)/\|{\mathbf a}\|_2^2$. Substituting $\hat{\mathbf x}$ in $S({\mathbf x},w)$ allows us to obtain the generalized likelihood ratio test (GLRT) estimate of $\omega$, which is the solution to the following optimization problem:
\begin{align}\label{cost}
\hat{\omega} = \underset{\omega}{\rm {argmax}}~{\mathbf G}_{\mathbf Y}(\omega)\triangleq \underset{\omega}{\rm {argmax}}~\sum\limits_{t=1}^T{\mathbf G}_{{\mathbf y}_t}(\omega),
\end{align}
where \begin{align}\label{sub}
{\mathbf G}_{{\mathbf y}_t}(\omega)= |{\mathbf y}_t^{\rm H}{\mathbf a}|^2/\Vert {\mathbf a}\Vert^2_{\rm{2}}.
\end{align}
is the GLRT cost function for the $t$th SMV. We use this to obtain a coarse estimate of $({\mathbf x},\omega)$ in the detection stage.


{\emph{Detection:}} By restricting $\omega$ to a finite discrete set denoted by $\Omega \triangleq \{k(2\pi/{\gamma}N):k=0, 1, \cdots, ({\gamma}N-1)\}$, where $\gamma$  is the oversampling factor relative to the DFT grid, we can obtain a coarse estimate of $\omega$. We treat the ${\omega}_c\in{\Omega}$ that maximizes the cost function (\ref{sub}) as the output of this stage, and the corresponding $\mathbf x$ vector estimate is ${\mathbf a}^{\rm H}(\omega_c){\mathbf y}_t/\Vert {{\mathbf a}(\omega_c)} \Vert^2_{\rm{2}}$.

{\emph{Refinement:}} Let $(\hat{\mathbf x}, \hat{\omega})$ denote the current estimate, then the Newton procedure for frequency refinement is
\begin{align}\label{rule}
\hat{\omega}' = \hat{\omega} - \dot{\mathbf G}_{\mathbf Y}(\hat{\omega})/\ddot{\mathbf G}_{\mathbf Y}(\hat{\omega}),
\end{align}
where $\dot{\mathbf G}_{\mathbf Y}(\hat{\omega})=\sum\limits_{t=1}^T\dot{\mathbf G}_{{\mathbf y}_t}(\hat{\omega})$ and $\ddot{\mathbf G}_{\mathbf Y}(\hat{\omega})=\sum\limits_{t=1}^T\ddot{\mathbf G}_{{\mathbf y}_t}(\hat{\omega})$ are simply summing over the $T$ corresponding the first and second order of the SMV terms, $\dot{\mathbf G}_{{\mathbf y}_t}(\hat{\omega})$ and $\ddot{\mathbf G}_{{\mathbf y}_t}(\hat{\omega})$ are given by \cite{Mamandipoor}
\begin{align}\label{des1}
\dot{\mathbf G}_{{\mathbf y}_t}(\hat{\omega}) = \Re\left\{\left({\mathbf y}_t - x_t{\mathbf a}(\hat \omega) \right)^{\rm H} x_t \left(d{\mathbf a}(\hat \omega)/d{\hat \omega} \right)\right\},
\end{align}
\begin{align}\label{des2}
\ddot{\mathbf G}_{{\mathbf y}_t}(\hat{\omega}) = \Re\left\{\left({\mathbf y}_t - x_t{\mathbf a}(\hat \omega) \right)^{\rm H} x_t \left(d^2{\mathbf a}(\hat \omega)/d{\hat \omega}^2 \right)\right\} - |x_t|^2\Vert d{\mathbf a}(\hat \omega)/d{\hat \omega}\Vert^2.
\end{align}
We maximize ${\mathbf G}_{\mathbf Y}({\omega})$ by employing the update rule (\ref{rule}) on the condition that the function is locally concave.
\subsection{Multiple frequency}
Assume that we have already detected {$L$} sinusoids, and let $P = \{({\mathbf x}_l,w_l), l=1,\cdots,L\}$ denote the set of estimates of the detected sinusoids. The residual measurement corresponding to this estimate is
\begin{align}
{\mathbf Y}_r(P)={\mathbf Y}-{\sum\limits_{l=1}^L {\mathbf a}(\omega_l){\mathbf x}_l^{\rm T}}
\end{align}
The method of estimating multiple frequencies proceeds by employing the single sinusoid procedure to perform Newtonized coordinate descent on the residual energy $\Vert{{\mathbf Y}_r(P)}\Vert_{\rm F}^{\rm 2}$. One step of this coordinate descent involves adjusting all $\omega_l$. The procedure to refine the $l$th sinusoid is as follows: ${\mathbf Y}_r(P\backslash \{{\mathbf x}_l, \omega_l\})$ now is referred to as the measurement ${\mathbf Y}$ and the single frequency update step is utilised to refine $({\mathbf x}_l, \omega_l)$.

Refinement Acceptance Condition (RAC): This refinement step is accepted when it results in a strict improvement in $G_{{\mathbf Y}_r(P\backslash \{{\mathbf x}_l, \omega_l\})}(\omega)$, namely, $G_{{\mathbf Y}_r}(\hat{\omega}')>G_{{\mathbf Y}_r}(\hat{\omega})$. By doing this, we can make sure that the adopted refinement must decrease the overall residual energy.

In summary, firstly, we detect a frequency $\hat \omega$ over the discrete set $\Omega$ by maximizing the cost function (\ref{sub}). Then we use the knowledge of the first-order and second-order derivative of the cost function to refine the estimate of $\hat \omega$. Next, we use the information of all the other previously detected sinusoids to further improve the estimation performance of every previously detected sinusoid one at a time. This step is crucial for the convergence and accuracy of the algorithm. Finally, we update $\mathbf x$ by least squares methods. The whole MNOMP is summarized in Algorithm \ref{Alg1}.

Here, we explain some main elements in MNOMP (Algorithm \ref{Alg1}):
\begin{itemize}
  \item SINGLE REFINEMENT (Step 7:) The single refinement locally refines the results obtained by coarsely picking the maximum over the dictionary matrix, and the number of single refinement is $R_s$.
  \item CYCLIC REFINEMENT (Step 9:) Through this process, a feedback is provided for local refinements of previously detected sinusoids, which allows us to better understand the received signal with the addition of another sinusoid. And the number of the cyclic refinement is $R_c$.
  \item UPDATE by least squares (Step 10:) By projecting the received signal onto the subspace spanned by the estimated frequencies, we update amplitudes of signals to make sure that the residual energy is the minimum possible for the present set of estimated frequencies.
\end{itemize}

Here, we compare the computational complexity of MNOMP and SPA \cite{Yang1}. As shown in Algorithm \ref{Alg1}, MNOMP includes IDENTITY STEP, SINGLE REFINEMENT STEP and CYCLIC STEP \cite{Mamandipoor}. According to \cite{Mamandipoor}, assuming that the proposed algorithm has run for precisely $K$ iterations, namely, we stop the algorithm when the model order of estimated signal is the same as that of the true signal. First, determining whether the stopping criterion is satisfied involves fast Fourier transform (FFT), with complexity $O(KNT{\rm {log}}(N))$. Second, as for Step $5$, the IDENTITY step, the GLRT cost function is calculated over the dictionary matrix, which can be implemented by using FFTs in $O(\gamma K N T {\rm {log}}(\gamma N))$. Then the SINGLE REFINEMENT step requires only $O(R_sN)$ operations per sinusoid per snapshot, hence the total cost for this step is $O(R_s KNT)$. Furthermore, the CYCLIC REFINEMENT refining all frequencies that have been estimated has complexity $O(R_c R_s K^2 N T)$. For SPA, its computation complexity is $O(N^2T+N^3+N^{6.5})$ \cite{Yang1}, which is higher than MNOMP.

\begin{algorithm}[ht]
\caption{MNOMP.}\label{Alg1}
1: \textbf{Procedure} EXTRACTSPECTRUM $({\mathbf Y, \tau}):$\\
2: $m\leftarrow 0$, ${P}_0 = \{\}$\\
3: $\textbf{while}$ {${\max}_{\omega\in {\rm {DFT}}}G_{{\mathbf Y}_r({P}_m)}(\omega)>\tau$}\\
4: $m\leftarrow m+1$\\
5: $\textbf{IDENTIFY}$\\
$\hat{\omega} = {\rm {arg~max}}_{\omega \in {\Omega}} G_{{\mathbf Y}_r({P}_{m-1})}(\omega)$\\
 and its corresponding ${\mathbf x}$ vector estimate \\$\hat {\mathbf x}^{\rm T} \leftarrow \left({\mathbf a}^{\rm H}(\hat \omega){\mathbf Y}_r({P}_{m-1})\right)/\Vert {{\mathbf a}(\hat \omega)} \Vert^2_{\rm{2}}$.
\\6: ${P_m'}\leftarrow {P}_{m-1}\cup \{(\hat {\mathbf x}, {\hat \omega})\}$
\\7: SINGLE REFINEMENT: Refine $(\hat {\mathbf x}, \hat \omega)$ using single frequency Newton update algorithm (${R}_s$ Newton steps) to obtain improved estimates $(\hat {\mathbf x}', \hat {\omega}')$.
\\8: ${P_m''}\leftarrow {P}_{m-1}\cup \{(\hat {\mathbf x}', {\hat \omega}')\}$
\\9: CYCLIC REFINEMENT: Refine parameters in ${P}''_m$ one at a time: For each $(\mathbf x, \omega)\in{P_m''}$, we treat ${\mathbf Y}_r({P_m''} \backslash \{(\mathbf x, \omega)\})$ as the measurement $\mathbf Y$, and apply single frequency Newton update algorithm. We perform ${R}_c$ rounds of cyclic refinements. Let ${P_m'''}$ denote the new set of parameters.
\\10: UPDATE all ${\mathbf x}$ vector estimate in ${P_m''}$ by least squares: $\mathbf A \triangleq [{\mathbf a}({\omega_1}),\cdots,{\mathbf a}(\omega_m)]$, $\{\omega_l\}$ are the frequencies in ${P_m'''}$. And $[{\mathbf x}_1,\cdots, {\mathbf x}_m]^{\rm T} = {\mathbf A}^\dagger {\mathbf Y}$. \\Let ${P}_m$ denote the new set of parameters.\\
11: $\textbf{end while}$
\\12: \textbf{return} $P_m$
\label{code:recentEnd}
\end{algorithm}

\section{Stopping criterion}
To understand the performance of MNOMP, the probability of the algorithm overestimating the model order $K$ is of interest. An extreme scenario is that the proposed algorithm has detected $K$ sinusoids which causes the residual to be only AWGN in model (\ref{multisnap}), and the stopping criterion still isn't met. So the algorithm has to detect another sinusoid to make the residual decrease, which corresponds to the scenario of overestimating the model order.

We use the stopping criterion to estimate the model order $K$. If the residual energy can be well explained by noise, up to a target overestimating probability, then we stop. Intuitively, we choose to terminate the algorithm by comparing the magnitude of the Fourier transform of the residual with the expected noise power. Details are given in the next section.
\subsection{Stopping criterion}
The algorithm stops when
\begin{align}
G_{{\mathbf Y}_r(P)}(\omega) =  {\sum\limits_{t=1}^T} \left|\left \langle {\mathbf y}_{rt}(P), {\mathbf a}(\omega)\right \rangle \right|^2 < \tau
\end{align}
for all DFT sampling frequencies $\{\omega_n\triangleq 2\pi n/N: n = 0, \cdots, N-1$\}, where ${\mathbf y}_{rt}(P)$ is the $t$th column of $\mathbf Y_{r}(P)$, $T$ is the number of snapshots and $\tau$ is the stopping threshold.

Supposedly, we have already correctly detected all sinusoids in the mixture. Under this condition, the residual is ${\mathbf y}_{rt}(P) \approx {\mathbf z_t}$, where ${\mathbf  z_t} \sim \mathcal{CN}(\mathbf 0, \sigma^2{\mathbf I_M})$. Then by defining $R_n \triangleq \sum_{t=1}^T| {\mathbf a}^{\rm H} (\omega_n){\mathbf z}_t|^2$, we obtain
\begin{align}\label{13}
&{\rm {Pr}}\left\{\underset{n=1, \cdots, N}{\rm {max}}~\sum_{t=1}^T| {\mathbf a}^{\rm H} (\omega_n){\mathbf z}_t |^2>\tau\right\}= {\rm {Pr}}\left( \underset{n=1, \cdots, N}{\rm {max}}~R_n>\tau\right)\notag\\
  =& 1-{\rm {Pr}}\left( \underset{n=1, \cdots, N}{\rm {max}}~R_n\leq\tau\right) = 1-{\rm {Pr}}\left( R_1\leq\tau,\cdots, R_N\leq \tau\right)
\end{align}
Note that by defining $u_{n,t} = {\mathbf a}^{\rm H} (\omega_n){\mathbf z}_t$, we have
\begin{align}\label{u}
&{\rm E}\left[u_{n_1,t_1} u^*_{n_2,t_2}\right]
= {\rm E}\left[{\mathbf a}^{\rm H} (\omega_{n_1}){\mathbf z}_{t_1}{\mathbf z}_{t_2}^{\rm H} {\mathbf a} (\omega_{n_2})\right]\notag\\
=& \sigma^2 \delta_{t_1, t_2} {\rm E}\left[ {\mathbf a}^{\rm H} (\omega_{n_1}){\mathbf a} (\omega_{n_2})\right] = \sigma^2 \delta_{t_1, t_2} \delta_{n_1, n_2},
\end{align}
where $\delta_{t_1, t_2}$ denotes the Dirac delta function, which equals zero unless $t_1 = t_2$ holds. From (\ref{u}), we can conclude that $R_n=\sum\limits_{t=1}^T|u_{n,t}|^2$ is a $\chi^2$ random variable with $2T$ degrees of freedom and common variance $\sigma^2/2$. With the degrees of freedom $2T$ being even and common variance $\sigma_0^2$, the cumulative distribution function (CDF) $F_{\chi_{2T}^2}(x,\sigma_0^2)$ has a closed form \cite[Equation (2.3-24)]{proagis}
\begin{align}
F_{\chi_{2T}^2}(x,\sigma_0^2)=
\begin{cases}
 & 1-{\rm e}^{-\frac{x}{2\sigma_0^2}}\sum_{k=0}^{T-1}\frac{1}{k!}\left(\frac{x}{2\sigma_0^2}\right)^k,~\tau>0 \\
 & 0, \quad{\rm otherwise}.
\end{cases}
\end{align}
Since $F_{\chi_{2T}^2}(x,\sigma_0^2)$ depends only on $x/\sigma_0^2$, we define $F_{\chi_{2T}^2}(x/\sigma_0^2)\triangleq F_{\chi_{2T}^2}(x,\sigma_0^2)$ for compactness. For our problem, (\ref{13}) can be calculated as
{\begin{align}
{\rm {Pr}}\left\{\underset{n=1, \cdots, N}{\rm {max}}~\sum_{t=1}^T|{\mathbf a}^{\rm H} (\omega_n){\mathbf z}_t |^2>\tau\right\}
 = 1 - \left[{\rm {Pr}}\left(R_n\leq\tau\right)\right]^N
 =1 -F_{\chi_{2T}^2}^N\left(\frac{2\tau}{\sigma^2}\right).\notag
\end{align}}
Let ${\rm P}_{\rm {oe}}$ denote a nominal overestimating probability. Thus ${\rm P}_{\rm {oe}}$ satisfies
{\begin{align}\label{poe}
{\rm {Pr}}\left\{\underset{n=1, \cdots, N}{\rm {max}}~\sum_{t=1}^T|{\mathbf a}^{\rm H} (\omega_n){\mathbf z}_t |^2> \tau\right\} = {\rm P}_{\rm {oe}}.
\end{align}}
By defining ${F}^{-1}_{\chi_{2T}^2}(\cdot)$ as the inverse function of $F_{\chi_{2T}^2}(\cdot)$, we obtain
\begin{align}\label{false}
\tau = {F}^{-1}_{\chi_{2T}^2}\left((1-{\rm P}_{\rm {oe}})^\frac{1}{N}\right)\sigma^2/2.
\end{align}
Note that for a single snapshot, i.e., $T=1$, ${F}_{\chi_{2}^2}(2{\tau}/{\sigma^2})=1-{\rm e}^{-{\tau}/{\sigma^2}}$ and $\tau=-\sigma^2{\rm {log}}\left(1-(1-P_{\rm {oe}})^{1/N} \right)$ from (\ref{false}), which is consistent with the results obtained in \cite{Mamandipoor}.

We conduct a numerical experiment by comparing the ``measured'' against ``nominal'' overestimating probability (\ref{poe}) to substantiate the above analysis. We use MNOMP to estimate frequencies in a mixture of $K = 16$ sinusoids of the same fixed nominal SNR, defined as
\begin{align}\label{snr}
{\rm SNR}_k = 10{\rm {log}_{10}}\left(\frac{\Vert{\mathbf a}(\omega_k){\mathbf x}_k^{\rm T}\Vert_{\rm F}^2}{\sigma^2T}\right)=10{\rm {log}_{10}}\left(\frac{\Vert{\mathbf x}_k\Vert_{\rm 2}^2}{\sigma^2T}\right), ~k = 1, \cdots, K.
\end{align}
Furthermore, we generate the frequencies such that the minimal wrap-around frequency separation is $\Delta \omega_{\rm {min}}$, where $\Delta \omega_{\rm {min}} = 2.5\times \Delta \omega_{\rm {DFT}}$ and $\Delta \omega_{\rm {DFT}}\triangleq 2\pi/N$ is the DFT grid separation. The parameters are set as follows: $N=256$, $K = 16$, $T=10$, $R_s = 1$, $R_c = 3$, the number of Monte Carlo (MC) trials is ${\rm MC}=300$. The ``measured'' overestimating probability is defined as the ratio of the overestimating events in all MC trials. Fig. \ref{Fig1} shows that the empirical overestimating probability is close to the nominal value at various SNRs, which means that the final residual error can be approximated as the AWGN, namely, the frequency estimation accuracy of MNOMP is good.
\begin{figure}[htbp]
 \centering
  \includegraphics[width=80mm]{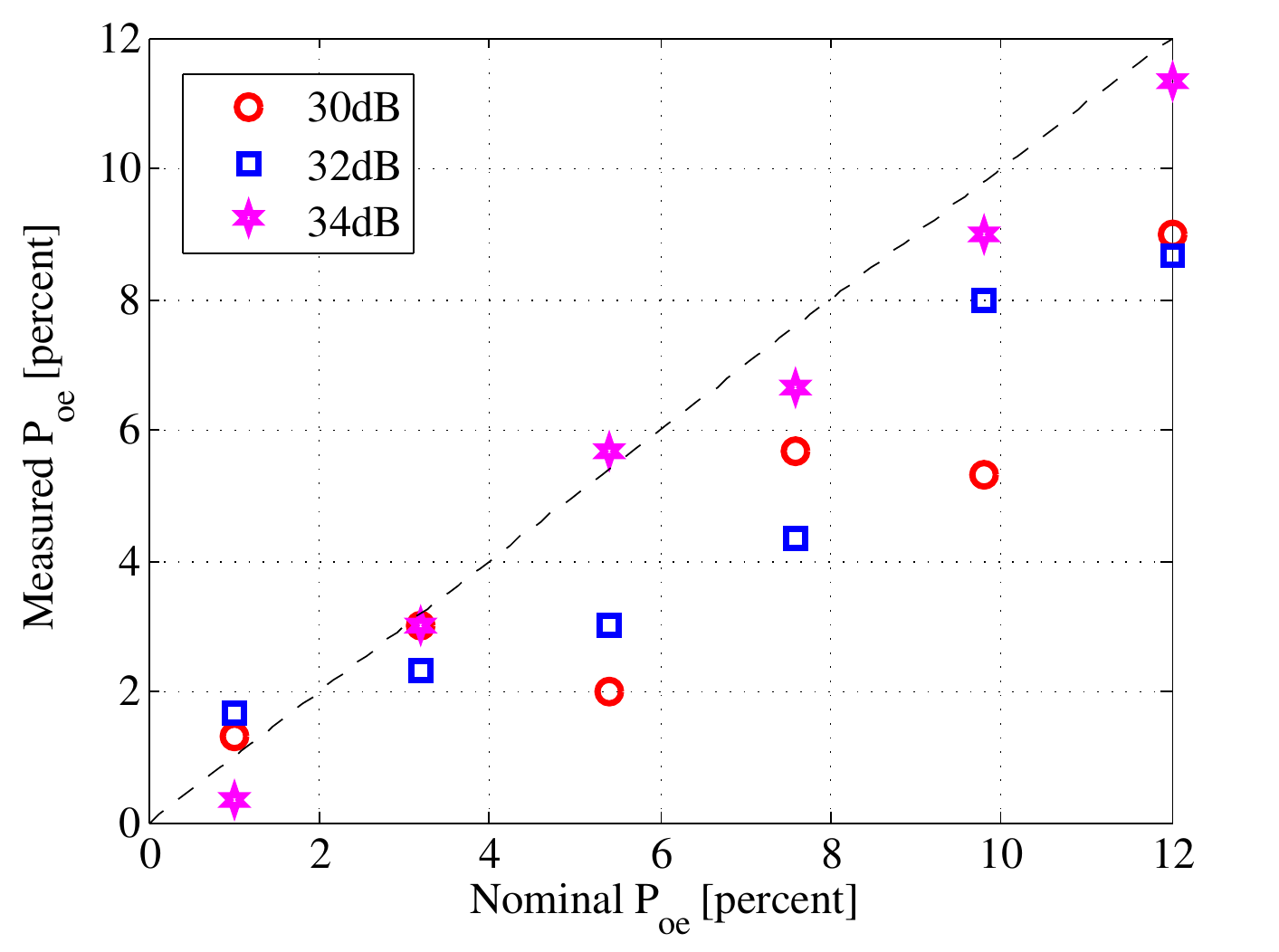}\\
  \caption{Nominal vs measured probability of overestimating probability.}\label{Fig1}
\end{figure}
\subsection{Probability of Miss}
Firstly, we define the neighborhood $N_{\omega_k}$ around the true frequency $\omega_k$ as $N_{\omega_k} \triangleq \{\omega:{\rm {dist}}(\omega, \omega_k)\leq 0.25\times \Delta\omega_{\rm {DFT}}\}$. Then we declare a miss of $\omega_k$ if none of the estimated frequencies locates in $N_{\omega_k}$, otherwise we declare a successful detection of $\omega_k$. Note that a miss is caused by both noise and inter-sinusoid interference, but we only discuss noise here.

Assuming no inter-sinusoid interference, the measurement $\mathbf Y$ can be described as ${\mathbf Y}=[{\mathbf a}(\omega)x_1 + {\mathbf z}_1, {\mathbf a}(\omega)x_2 +{\mathbf z}_2, \cdots, {\mathbf a}(\omega) x_T+{\mathbf z}_T]$. The signal frequency $\omega$ is not detected if
\begin{align}
{\rm P}_{\rm {miss}} = {\rm {Pr}}\left\{\Vert {{\mathbf a}^{\rm H}({\omega}_d){\mathbf Y}} \Vert^2_2< \tau\right\},
\end{align}
where $\omega_d$ denotes the sampling frequency. Note that ${x_i}$ is deterministic known, and ${\mathbf z}_i\sim{\mathcal {CN}}(0, \sigma^2{\mathbf I}_N)$. Hence, we have
\begin{align}
{\rm P}_{\rm {miss}} = {\rm {Pr}}\left\{  \sum_{t=1}^T\left| {{\mathbf a}^{\rm H}({\omega}_d)\left[{\mathbf a}(\omega)x_t+{\mathbf z}_t\right]} \right|^2< \tau\right\}.
\end{align}
To calculate the above probability, we introduce
\begin{align}
&{v_t}\triangleq {{\mathbf a}^{\rm H}({\omega}_d)\left[{\mathbf a}(\omega)x_t+{\mathbf z}_t\right]}
=x_t e^{{\rm j}\frac{(N-1)(\omega-\omega_d)}{2}}\frac{{\rm {sin}}\left(\frac{N(\omega-\omega_d)}{2}\right)}{N{\rm {sin}}(\frac{\omega-\omega_d}{2})}+{\mathbf a}^{\rm H}(\omega_d){\mathbf z_t}.
\end{align}
By defining
\begin{align}
\alpha \triangleq \frac{   {\rm {sin}}\left[ N(\omega - \omega_d)/2 \right]    }{N {\rm {sin}}\left[ (\omega - \omega_d)/2\right]},\quad{\tilde R} \triangleq \Vert {{\mathbf a}^{\rm H}({\omega}_d){\mathbf Y}} \Vert^2_2=\sum_{t=1}^T|{v_t}|^2,
\end{align}
we can conclude that $\tilde{R}$ is a noncentral $\chi^2$ random variable with $2T$ degrees of freedom and common variance being $\sigma^2/2$. With degrees of freedom $2T$ being an even number, the CDF of $\tilde{R}$ can be written in the form
\begin{align}
{F}_{2T}(\tau)=\begin{cases}
 & 1-{\rm Q}_T\left( \frac{\sqrt{2}s}{\sigma},~\frac{\sqrt{2\tau}}{\sigma}\right),~\tau>0 \\
 & 0,
\end{cases}
\end{align}
where ${\rm Q}_T$ denotes Marcum ${\rm Q}$-function and {the noncentral parameter} $s$ is defined as
\begin{align}
s \triangleq \alpha\sqrt{\sum_{t=1}^{T}|x_t|^2}.
\end{align}
Thus it's easy to show that
\begin{align}
{\rm P}_{\rm {miss}} = 1 - {\rm Q}_T\left(\alpha{\sqrt{2\sum_{t=1}^T|{x_t|^2}}}/{\sigma}, \sqrt{2\tau/\sigma^2}\right).
\end{align}

Supposing a frequency within a DFT grid interval follows the uniform distribution, then we obtain ${\rm E}[\alpha] = 0.88$, where $\omega\sim {\rm {U}}[-\pi/N, \pi/N]$ \cite{Mamandipoor}. Hence,
\begin{align}
{\rm P}_{\rm {miss}} = 1-Q_T\left(0.88{\sqrt{2\sum_{t=1}^T|{x_t|^2}/{\sigma^2}}}, \sqrt{2\tau/\sigma^2}\right).
\end{align}
According to the definition of SNR (\ref{snr}), we have
\begin{align}\label{mis}
{\rm P}_{\rm {miss}} = 1-Q_T\left(0.88{\sqrt{2T{\rm {SNR}}}}, \sqrt{2\tau/\sigma^2}\right).
\end{align}

Here the probability of miss is linked to the probability of overestimating via the threshold $\tau$. It is meaningful to analyze the effects of snapshots $T$. For the generalized Marcum Q-function $Q_v(a, b)$, it is shown that it is strictly increasing in $v$ and $a$ for all $a\geq 0$ and $b, v > 0$, and is strictly decreasing in $b$ for all $a,b\geq 0$ and $v>0$ \cite{Sun}. In our case, $v = T$, $a = 0.88{\sqrt{2T{\rm {SNR}}}}$, $b = \sqrt{2\tau/\sigma^2}$. Let $T_1 > T_2$. Fixing the SNR and ${\rm P}_{\rm oe}$, according to equation (\ref{false}) the threshold $\tau$ depends on $T$ and $\tau_1 > \tau_2$. Obviously we have $v_1 > v_2$, $a_1>a_2$ and $b_1>b_2$. Thus it is difficult to obtain whether ${\rm P}_{\rm miss}|_{T=T_1}>{\rm P}_{\rm miss}|_{T=T_2}$ or not.

A simple simulation is conducted to plot ${\rm P}_{\rm miss}$ versus ${\rm P}_{\rm oe}$ under different snapshots. The parameters are the same as Fig. \ref{Fig1} except that $K=8$, ${\rm MC} = 500$ and ${\rm SNR} = 11$ dB. The results are shown in Fig. \ref{roc}. Note that the computed result does not approximate well with the measured results under $T=1$ and the other parameter settings. As snapshots increases, the computed results become more accurate. With the probability of overestimating ${\rm P}_{\rm oe}$ being fixed, the probability of miss ${\rm P}_{\rm miss}$ decreases as snapshots increases. This demonstrates that MNOMP benefits from MMVs.
\begin{figure}[htbp]
 \centering
  \includegraphics[width=80mm]{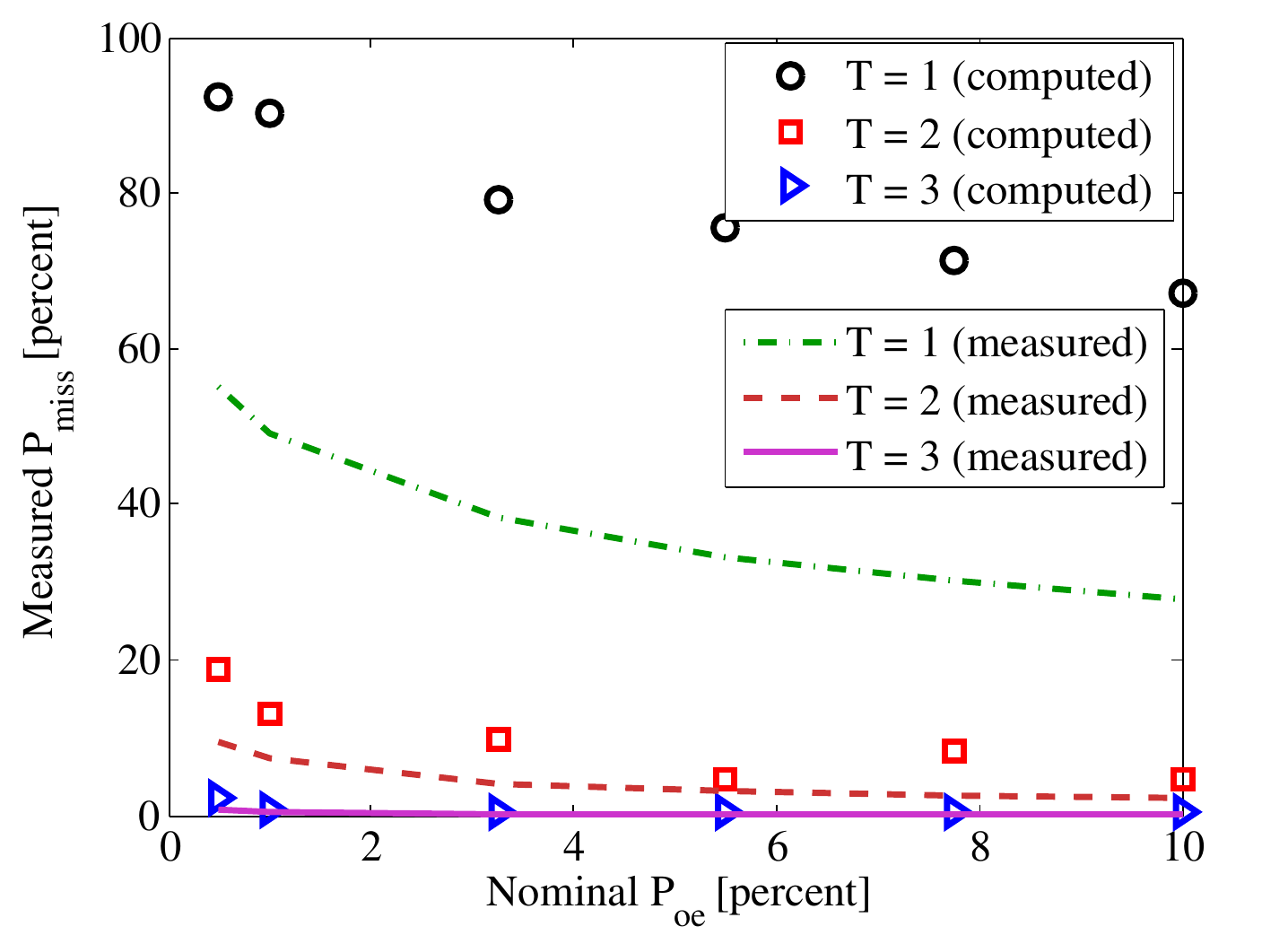}\\
  \caption{Probability of miss $P_{\rm {miss}}$ versus probability of overestimating $P_{\rm {oe}}$.p}\label{roc}
\end{figure}
\section{Convergence}
In this section, the convergence of the proposed algorithm is studied. Firstly, upper bounds on the number of iterations needed to reach the stopping condition are given. Then a bound on the rate of convergence of MNOMP is provided, which is a function of the ``atomic norm'' of original measurements $\mathbf Y$ and the oversampling factor $\gamma$.
\subsection{Proof of convergence}
It's easy to show that the number of measurements $N$ is a trivial upper bound of the number of iterations of MNOMP. From Update step $10$ of MNOMP, it can be shown that $\mathbf X$ becomes a square full-rank matrix (note that there is no frequency that will be detected twice) after $N$ iterations. Then the algorithm terminates because the residual is equal to zero for the $(N +1)$th iteration.

In the following, we provide another upper bound on the number of iterations, which is obtained by considering how much the residual energy will decrease when a new frequency is added to the set of estimated sinusoids.
\begin{theorem}
Let $\Vert {\mathbf Y} \Vert_{\rm F}^2$ be the original residual energy. When a new frequency is added to the set of estimated sinusoids, the reduction of residual energy is at least $\tau$. Consequently, ${\rm {min}}\left\{  N, \lfloor\Vert {\mathbf Y} \Vert_{\rm F}^2/\tau\rfloor\right\}$ is an upper bound on the number of iterations of MNOMP.
\end{theorem}
\begin{proof}
Assume we have detected $m$ sinusoids, the residual measurement is given by $\mathbf Y_{\rm r}(P_m) = {\mathbf Y} - \sum_{l=1}^m{\mathbf a}(\omega_l){\mathbf x}_l^{\rm T}$. The residual energy in each iteration of MNOMP satisfies the following,
\begin{align}
\Vert {{\mathbf Y_r}(P_{m-1})} \Vert^2_{\rm F}
&\overset{(a)}=\Vert {{\mathbf Y_r}(P'_m)} \Vert^2_{\rm F} + G_{\mathbf Y_r(P_{m-1})}(\hat {\omega})\overset{(b)}\geq\Vert {{\mathbf Y_r}(P'''_m)} \Vert^2_{\rm F} + G_{\mathbf Y_r(P_{m-1})}(\hat {\omega})\notag\\
&\overset{(c)}\geq\Vert {{\mathbf Y_r}(P_m)} \Vert^2_{\rm F} + G_{\mathbf Y_r(P_{m-1})}(\hat {\omega})\label{16}\\
&\overset{(d)}\geq\Vert {{\mathbf Y_r}(P_m)} \Vert^2_{\rm F} + \tau.\label{17}
\end{align}
where $\hat \omega$ denotes the detected frequency. Equality in $(a)$ holds because of the Step 5 in MNOMP where we project ${\mathbf Y}_r(P_{m-1})$ orthogonal to the subspace spanned by $\mathbf a(\hat \omega)$ to get ${\mathbf Y}_r(P'_m)$. We obtain $\left[  {\mathbf Y_r}(P_{m-1})-{\mathbf a}(\hat \omega){\mathbf x}_m^{\rm T}\right]^{\rm H}{\mathbf a}(\hat \omega)={\mathbf 0}$ and ${\mathbf x}_m^{\rm T} = \left[{\mathbf a}(\hat \omega)^{\rm H}{\mathbf Y_r}(P_{m-1})\right]/\Vert {\mathbf a}(\hat \omega)\Vert^2_{\rm{2}}$ by solving $\underset{\mathbf x_m}{\rm {min}}~\Vert {\mathbf Y_r}(P_{m-1})-{\mathbf a}(\hat \omega){\mathbf x}_m^{\rm T} \Vert_{\rm F}^2$. Then we take the Frobenius norm of both sides of ${\mathbf Y}_r(P'_m)={\mathbf Y}_r(P_{m-1}) - {\mathbf a}(\hat \omega){\mathbf x}_m^{\rm T}$ to obtain equality (a). Inequality in (b) follows from RAC, which is implemented whenever the Single Refinement step is executed and (c) is a direct consequence of the Step 10 of MNOMP, which can only cause a decrease in the residual energy. And the stopping criterion leads to (d).

From inequality (\ref{17}), we can conclude that the reduction of the residual energy caused by the detection of a new sinusoid frequency is always larger than $\tau$. Thus we can get another bound on the number of iterations of MNOMP, which is $\lfloor\Vert {\mathbf Y} \Vert_{\rm F}^2/\tau\rfloor$.
\end{proof}

\subsection{Rate of convergence}
In the noiseless scenario, the noiseless measurement ${\mathbf Y}^o$ can be written as
\begin{align}\label{model}
{\mathbf Y}^o = \sum_{k=1}^K{\mathbf a}(\omega_k){\mathbf x_k^{\rm T}}=\sum_{k=1}^K c_k {\mathbf a}(\omega_k) \phi_k^{\rm T}.
\end{align}
where $c_k = \Vert  \mathbf x_k\Vert_2>0$, $\omega_k\in[0,2\pi)$, $\phi_k = c_k^{-1}\mathbf x_k\in{\mathbb C}^{T\times 1}$ with $\Vert {\phi_k} \Vert_2=1$. Then the continuous dictionary or the set of atoms is given as
\begin{align}
{\mathcal A}\triangleq\left\{{\mathbf A}(\omega, \phi)={\mathbf a}(\omega)\phi^{\rm T} : \omega\in[0, 2\pi), \Vert \phi \Vert_2 = 1\right\}.
\end{align}

It is easy to see that $\mathbf Y^o$ is a linear combination of many atoms in $\mathcal A$. Here, we define the atomic $\ell_0$ (pseudo-) norm of $\mathbf Y\in{\mathbb C}^{N\times T}$ as the smallest number of atoms that can express it \cite{Chandrasekaran}:
\begin{align}
\Vert{\mathbf Y}\Vert_{\mathcal A, 0}=\underset{{\mathcal K}}{\rm {inf}}\left\{\mathbf Y=\sum_{k=1}^{\mathcal K}c_k{\mathbf A}(\omega_k, \phi_k), c_k\geq 0  \right\}.
\end{align}
Because the atomic $\ell_0$ norm is non-convex, we utilize convex relaxation to relax the atomic $\ell_0$ norm to the atomic norm, which is defined as the gauge function of $\rm {conv} (\mathcal A)$, where $\rm {conv} (\mathcal A)$ is the convex hull of $\mathcal A$ \cite{Chandrasekaran}:
\begin{align}
&\Vert{\mathbf Y}\Vert_{\mathcal A}
\triangleq{\rm {inf}}\{t>0:\mathbf Y\in t\rm {conv} (\mathcal A)\}\notag\\
=&{\rm {inf}}\left\{\sum_k c_k:\mathbf Y=\sum_{k}c_k{\mathbf A}(\omega_k, \phi_k),  c_k\geq0\right\}.
\end{align}
By defining $\langle \mathbf Y, \mathbf A \rangle={\rm {tr}}({\mathbf A}^{\rm H}{\mathbf Y})$ and $\langle \mathbf Y, \mathbf A\rangle_{\mathbb R} = \mathscr R(\langle \mathbf Y, \mathbf A \rangle)$, the dual norm of $\Vert{\mathbf Y}\Vert_{\mathcal A}$ can be written as
\begin{align}\label{dual}
&\Vert{\mathbf Y}\Vert_{\mathcal A}^{*}
\triangleq\underset{\Vert{\mathbf A}\Vert_{\mathcal A}\leq 1}{\rm {sup}}~\langle \mathbf Y, \mathbf A\rangle_{\mathbb R}=\underset{{\omega\in [0,2\pi)}, \Vert \phi \Vert_2 = 1}{\rm {sup}}~\langle \mathbf Y, {\mathbf a(\omega)}{\phi}^{\rm T}\rangle_{\mathbb R}\notag\\
=&\underset{{\omega\in [0,2\pi)}}{\rm {sup}}~\Vert{\mathbf Y}^{\rm H}{\mathbf a}(\omega) \Vert_2=\underset{\omega\in [0,2\pi)}{\rm {sup}} ~\sqrt{G_{\mathbf Y}(\omega)},
\end{align}
where the ${G_{\mathbf Y}(\omega)}$ is (\ref{sub}).
\begin{theorem}\label{theorem2}
Maximizing $G_{\mathbf Y}(\omega)$ (for the dictionary of unit norm sinusoids) over $[0,2\pi)$ is consistent with that over the oversampled grid $\Omega$ with oversampling factor $\gamma$. Namely, we have
\begin{align}\label{T2}
&\underset{\omega\in \Omega}{\rm {max}}~\sqrt{G_{\mathbf Y}(\omega)}\leq \underset{\omega\in [0,2\pi)}{\rm {sup}}~\sqrt{G_{\mathbf Y}(\omega)}\leq \left(1-\frac{4\pi T}{\gamma}\right)^{-1/2}~\underset{\omega\in \Omega}{\rm {max}}~\sqrt{G_{\mathbf Y}(\omega)}
\end{align}
\end{theorem}
\begin{proof}
According to (\ref{dual}), $\left(\Vert{\mathbf Y}\Vert_{\mathcal A}^{*}\right)^2$ can be expressed as
\begin{align}\label{31}
\left(\Vert{\mathbf Y}\Vert_{\mathcal A}^{*}\right)^2
&= \underset{{\omega\in [0,2\pi)}}{\rm {sup}}~\Vert{\mathbf Y}^{\rm H}{\mathbf a}(\omega) \Vert_2^2=\underset{{\omega\in [0,2\pi)}}{\rm {sup}}~\sum_{t=1}^T\left|\frac{1}{\sqrt{N}}\sum_{n=1}^N \phi_{n,t}^*e^{j(n-1)\omega}  \right|^2\notag\\
&=\underset{{\omega\in [0,2\pi)}}{\rm {sup}}~\sum_{t=1}^T\left|W_t(\omega) \right|^2,
\end{align}
where $\phi_{n,t}$ is the $(n,t)$th entry of $\mathbf Y$, and $W_t(\omega)\triangleq \frac{1}{\sqrt{N}}\sum_{n=1}^N\phi_{n,t}^*e^{j(n-1){\omega}}$. For $\omega_1, \omega_2\in [0,2\pi)$, according to Bernstein's theorem \cite{Schaeffer} and the result in \cite[Appendix C]{Bhaskar}, we can obtain
\begin{align}\label{34}
|W_t(\omega_1)| - |W_t(\omega_2)| \leq 2 N|\omega_1-\omega_2|\underset{{\omega\in [0,2\pi)}}{\rm {sup}}~|W_t(\omega)|.
\end{align}
Then we have
\begin{small}
\begin{align}
&\sum_{t=1}^T|W_t(\omega_1)|^2 - \sum_{t=1}^T|W_t(\omega_2)|^2\overset{(a)}\leq \sum_{t=1}^T(|W_t(\omega_1)| + |W_t(\omega_2)|)\times 2 N |\omega_1 - \omega_2|\notag\\
&\times \underset{{\omega\in [0,2\pi)}}{\rm {sup}}~|W_t(\omega)|\overset{(b)}\leq 4 N  |\omega_1 - \omega_2| T \underset{{\omega\in [0,2\pi)}}{\rm {sup}}~\sum_{t=1}^T|W_t(\omega)|^2\label{36}=4 N  |\omega_1 - \omega_2|  T \left(\Vert{\mathbf Y}\Vert_{\mathcal A}^{*}\right)^2 ,\notag
\end{align}
\end{small}
where inequality (a) follows from (\ref{34}), inequality (b) is from the fact
\begin{align}
\sum_{t=1}^T\underset{{\omega\in [0,2\pi)}}{\rm {sup}}~|W_t(\omega)|^2\leq T\underset{{\omega\in [0,2\pi)}}{\rm {sup}}~\sum_{t=1}^T |W_t(\omega)|^2,
\end{align}
and the last equality follows by (\ref{31}).

Let $\omega_2$ take any value of the grid points $\left\{0,\frac{2\pi}{\gamma N},\cdots, \frac{2\pi(\gamma N - 1)}{\gamma N}\right\}$, we have
\begin{align}
\left(\Vert{\mathbf Y}\Vert_{\mathcal A}^{*}\right)^2
&=\underset{{\omega\in [0,2\pi)}}{\rm {sup}}~\sum_{t=1}^T\left|W_t(\omega) \right|^2=\underset{{\omega\in [0,2\pi)}}{\rm {sup}}~\sum_{t=1}^T\left(\left|W_t(\omega) \right|^2 -\left|W_t\left(\widetilde \omega\right)\right|^2+\left|W_t\left( \widetilde \omega\right)\right|^2\right)\notag\\
&\leq \underset{{\omega\in [0,2\pi)}}{\rm {sup}}~\sum_{t=1}^T\left(\left|W_t(\omega) \right|^2 -\left|W_t\left( \widetilde \omega\right)\right|^2\right)+\underset{ d=0,\cdots,\gamma N-1}{\rm {max}}~\sum_{t=1}^T \left|W_t\left(\frac{2\pi d}{\gamma N}\right)\right|^2\notag\\
& \leq 4N |\omega - \widetilde \omega|T \left(\Vert{\mathbf Y}\Vert_{\mathcal A}^{*}\right)^2 + \underset{ d=0,\cdots,\gamma N-1}{\rm {max}}~\sum_{t=1}^T \left|W_t\left(\frac{2\pi d}{\gamma N}\right)\right|^2\notag\\
&\leq 4N\frac{2\pi}{2\gamma N}T\left(\Vert{\mathbf Y}\Vert_{\mathcal A}^{*}\right)^2 + \underset{ d=0,\cdots,\gamma N-1}{\rm {max}}~\sum_{t=1}^T \left|W_t\left(\frac{2\pi d}{\gamma N}\right)\right|^2\notag\\
&= \frac{4\pi T}{\gamma} \left(\Vert{\mathbf Y}\Vert_{\mathcal A}^{*}\right)^2+ \underset{ d=0,\cdots,\gamma N-1}{\rm {max}}~\sum_{t=1}^T \left|W_t\left(\frac{2\pi d}{\gamma N}\right)\right|^2
\end{align}
Since the maximum on the grid is a lower bound for $\left(\Vert{\mathbf Y}\Vert_{\mathcal A}^{*}\right)^2$, we have
\begin{align}
& \left( \underset{d=0,\cdots,\gamma N-1}{\rm {max}}~\sum_{t=1}^T \left|W_t\left(\frac{2\pi d}{\gamma N}\right)\right|^2\right)^{1/2}
 \leq \Vert{\mathbf Y}\Vert_{\mathcal A}^{*}\notag\\
& \leq \left(1-\frac{4\pi T}{\gamma}\right)^{-1/2}\left( \underset{d=0,\cdots,\gamma N-1}{\rm {max}}~\sum_{t=1}^T \left|W_t\left(\frac{2\pi d}{\gamma N}\right)\right|^2\right)^{1/2}
 \end{align}
 Thus,
 \begin{align}
 &\underset{\omega\in \Omega}{\rm {max}}~\sqrt{G_{\mathbf Y}(\omega)}\leq \underset{\omega\in [0,2\pi)}{\rm {sup}}~\sqrt{G_{\mathbf Y}(\omega)} \leq \left(1-\frac{4\pi T}{\gamma}\right)^{-1/2}~\underset{\omega\in \Omega}{\rm {max}}~\sqrt{G_{\mathbf Y}(\omega)}.
\end{align}
\end{proof}

\begin{remark}
Theorem 2 shows that the $\underset{\omega\in [0,2\pi)}{\rm {sup}}~\sqrt{G_{\mathbf Y}(\omega)}$ can be upper bounded by $\underset{\omega\in \Omega}{\rm {max}}~\sqrt{G_{\mathbf Y}(\omega)}$ which takes values at discreet grids within a scale factor $\left(1-\frac{4\pi T}{\gamma}\right)^{-1/2}$. It can be seen that the inequality is tight when $\gamma\rightarrow \infty $. To ensure that the scale factor is constant, the oversampling rate $\gamma$ must increase linearly with the number of snapshots $T$.
\end{remark}

To prove Theorem 3, the following lemma  in \cite{Barron} is introduced
\begin{lemma}\label{lemma1}
\cite{Barron} Assume $\{a_n\}_{n\geq 0}$ is a decreasing sequence of nonnegative numbers such that $a_0\leq U$ and
\begin{align}
a_n\leq a_{n-1}\left(1-\frac{a_{n-1}}{U} \right), \forall n>0,
\end{align}
then we have $a_n\leq \frac{U}{n+1}$ for all $n\geq 0$.
\end{lemma}
\begin{theorem}\label{theorem3}
For all $\mathbf Y$ satisfying $\Vert {\mathbf Y} \Vert_{\mathcal A} < \infty$, the residual energy of MNOMP at the $m$th iteration satisfies
\begin{align}
\Vert {\mathbf Y}_r(P_m)\Vert_{\rm F} \leq (m+1)^{-1/2}\left(1-\frac{4\pi T}{\gamma}\right)^{-1/2}\Vert {\mathbf Y}\Vert_{\mathcal A}.
\end{align}
\end{theorem}
\begin{proof}
From (\ref{16}), we have
\begin{align}\label{25}
\Vert {{\mathbf Y_r}(P_m)} \Vert^2_{\rm F} \leq \Vert {{\mathbf Y_r}(P_{m-1})} \Vert^2_{\rm F} - G_{\mathbf Y_r(P_{m-1})}(\hat {\omega}).
\end{align}
${\mathbf Y}_r(P_{m-1})$ is a direct consequence of projecting $\mathbf Y$ orthogonal to the subspace spanned by $P_{m-1}$, therefore
\begin{align}
&\Vert {{\mathbf Y_r}(P_{m-1})} \Vert_{\rm F}^2
=\langle{\mathbf Y}_r(P_{m-1}), \mathbf Y  \rangle_{\mathbb R}\overset{(a)}\leq \Vert {{\mathbf Y}} \Vert_{\mathcal A}  \Vert {{\mathbf Y}_r(P_{m-1})} \Vert^*_{\mathcal A}\notag\\
=&\Vert {{\mathbf Y}} \Vert_{\mathcal A}  \underset{\omega\in [0,2\pi)}{\rm {sup}}~\sqrt{G_{\mathbf Y_r(P_{m-1})}({\omega})}\overset{(b)} \leq \Vert {{\mathbf Y}} \Vert_{\mathcal A}  \left(1-\frac{4\pi T}{\gamma}\right)^{-1/2} \underset{\omega\in \Omega}{\rm {max}}\sqrt{G_{\mathbf Y_r(P_{m-1})}({\omega})},\notag
\end{align}
where inequality (a) follows from the following
\begin{align}
\Vert {{\mathbf Y_r}(P_{m-1})} \Vert_{\mathcal A}^*
&= \underset{\Vert \mathbf Y\Vert_{\mathcal A}\leq1}{\rm {sup}}~\langle\mathbf Y,{\mathbf Y}_r(P_{m-1})  \rangle_{\mathbb R}\geq \left\langle \frac{\mathbf Y}{\Vert \mathbf Y\Vert_{\mathcal A}},{\mathbf Y}_r(P_{m-1}) \right \rangle_{\mathbb R}
\end{align}
and equality (b) is obtained from Theorem \ref{theorem2}. From the step 5 of MNOMP, we have
\begin{align}\label{omega}
\hat {\omega} =  \underset{\omega\in \Omega}{\rm {arg~max}} \sqrt{G_{\mathbf Y_r(P_{m-1})}({\omega})},
\end{align}
By defining $\eta\triangleq \Vert {{\mathbf Y}} \Vert_{\mathcal A}\left(1-\frac{4\pi T}{\gamma}\right)^{-1/2}$ and combing with (\ref{omega}), we have
\begin{align}\label{26}
\Vert {{\mathbf Y_r}(P_{m-1})} \Vert_{\rm F}^2 \leq \eta  \sqrt{G_{\mathbf Y_r(P_{m-1})}(\hat {\omega})}.
\end{align}
Combining (\ref{25}) and (\ref{26}), yields
\begin{align}
\Vert {{\mathbf Y_r}(P_{m})} \Vert_{\rm F}^2 \leq \Vert {{\mathbf Y_r}(P_{m-1})} \Vert_{\rm F}^2  \left(1-\eta^{-2} \Vert {{\mathbf Y_r}(P_{m-1})} \Vert_{\rm F}^2 \right).
\end{align}
By using Lemma \ref{lemma1} and the fact that
\begin{align}
\Vert {{\mathbf Y_r}(P_{0})} \Vert_{\rm F}^2 = \Vert {\mathbf Y} \Vert_{\rm F}^2\overset{(a)}\leq \Vert {\mathbf Y} \Vert_{\mathcal A}^2\leq \eta^2,
\end{align}
where inequality $(a)$ follows from
\begin{align}\label{51}
 \Vert {\mathbf Y} \Vert_{\rm F}^2&=\sum_{k=1}^K\sum_{l=1}^K \left({\mathbf x}_k^{\rm T}{\mathbf x}_l^*\right)\left( {\mathbf a}^{\rm H}(\omega_l) {\mathbf a}(\omega_k)\right) \leq \sum_{k=1}^K\sum_{l=1}^K |{\mathbf x}_l^{\rm H}{\mathbf x}_k| |{\mathbf a}^{\rm H}(\omega_k) {\mathbf a}(\omega_l)| \notag\\
 &\leq \sum_{k=1}^K\sum_{l=1}^K |{\mathbf x}_l^{\rm H}{\mathbf x}_k| \leq \sum_{k=1}^K\sum_{l=1}^K \Vert\mathbf x_k \Vert_2\Vert\mathbf x_l \Vert_2 =\Vert {\mathbf Y} \Vert_{\mathcal A}^2
\end{align}
we have
\begin{align}
\Vert {{\mathbf Y_r}(P_{m})} \Vert_{\rm F}^2\leq \frac{\eta^2}{m+1}.
\end{align}
In other words,
\begin{align}\label{conver}
\Vert {{\mathbf Y_r}(P_{m})} \Vert_{\rm F}\leq (m+1)^{-1/2}\left(1-\frac{4\pi T}{\gamma}\right)^{-1/2}\Vert {{\mathbf Y}} \Vert_{\mathcal A }.
\end{align}
This proves Theorem \ref{theorem3}.
\end{proof}


Note that the component in (\ref{conver}) is $\left(1-{4\pi T}/{\gamma}\right)^{-1/2}$ instead of $\left(1-{2\pi }/{\gamma}\right)^{-1}$  in \cite{Mamandipoor}. Given $T = 1$, our conclusion isn't consistent with that in \cite{Mamandipoor}. In fact, due to $\left(1-{4\pi T}/{\gamma}\right)^{-1/2} \geq \left(1-{2\pi }/{\gamma}\right)^{-1}$, the bound on the rate of convergence is tighter in the single snapshot scenario. The reason is that the bound from Theorem 2 is looser than that of \cite[Theorem 2]{Mamandipoor}.

\begin{remark}
For the above convergence analysis of the proposed MNOMP algorithm, no separation condition is needed. It is worth noting that, no separation condition is also imposed for the analysis of ${\rm P}_{\rm oe}$ and ${\rm P}_{\rm miss}$. Nevertheless, as shown in Section IV, to obtain the probability of overestimating, the residue is approximated as the noise (${\mathbf y}_{rt}(P) \approx {\mathbf z_t}$). This means that MNOMP detects all the true sinusoids approximately, which holds true with high probability if the frequencies are well separated. For the probability of miss, we claim that there is no inter-sinusoidal interference. Such a condition also implies that the frequencies are well separated.
\end{remark}
\subsection{Empirical rate of convergence}
The relative residual energy of the $i$th iteration (averaged over 300 runs) versus the number of iterations in a noiseless scenario is plotted to show the effects of the refinement steps on the convergence of MNOMP, defined as $20{\rm {log}_{10}}\left( \Vert {{\mathbf Y_r}(P_{m})} \Vert_{\rm F}/ \Vert {{\mathbf Y}} \Vert_{\mathcal A }\right)$. Here, we compare MNOMP in the scenarios of $T = 1$ amd $T = 10$ to the following variants of OMP to show the improvements brought by the refinement steps.

\emph{MNOMP without cyclic refinements (MNOMP-)}: This is an algorithm that has a nearly comparable performance with OMP \cite{Tropp111} over the continuum of atoms. We use MNOMP to implement this method by setting the number of Cyclic Refinement Steps to 0. This algorithm lies in the class of forward greedy methods because it doesn't have a feedback mechanism. Hence, our analysis is also applicable to this method.

\emph{Multi-snapshot Discretized OMP (MDOMP)}: If we skip the Single and Cyclic Refinement Steps, then we can obtain a standard OMP applied to the oversampled grid $\Omega$. Note that this algorithm can be viewed as a special case of MNOMP, so the convergence analysis results are also valid.

The parameters are set as follows: $K=16$, $N=64$, $T=10$, $R_s = 1$, $R_c = 1$, ${\rm {SNR}}=25~{\rm dB}$, $\Delta \omega_{\rm {min}}=2.5\Delta \omega_{\rm {DFT}}$.
\begin{figure}
 \centering
  \includegraphics[width=80mm]{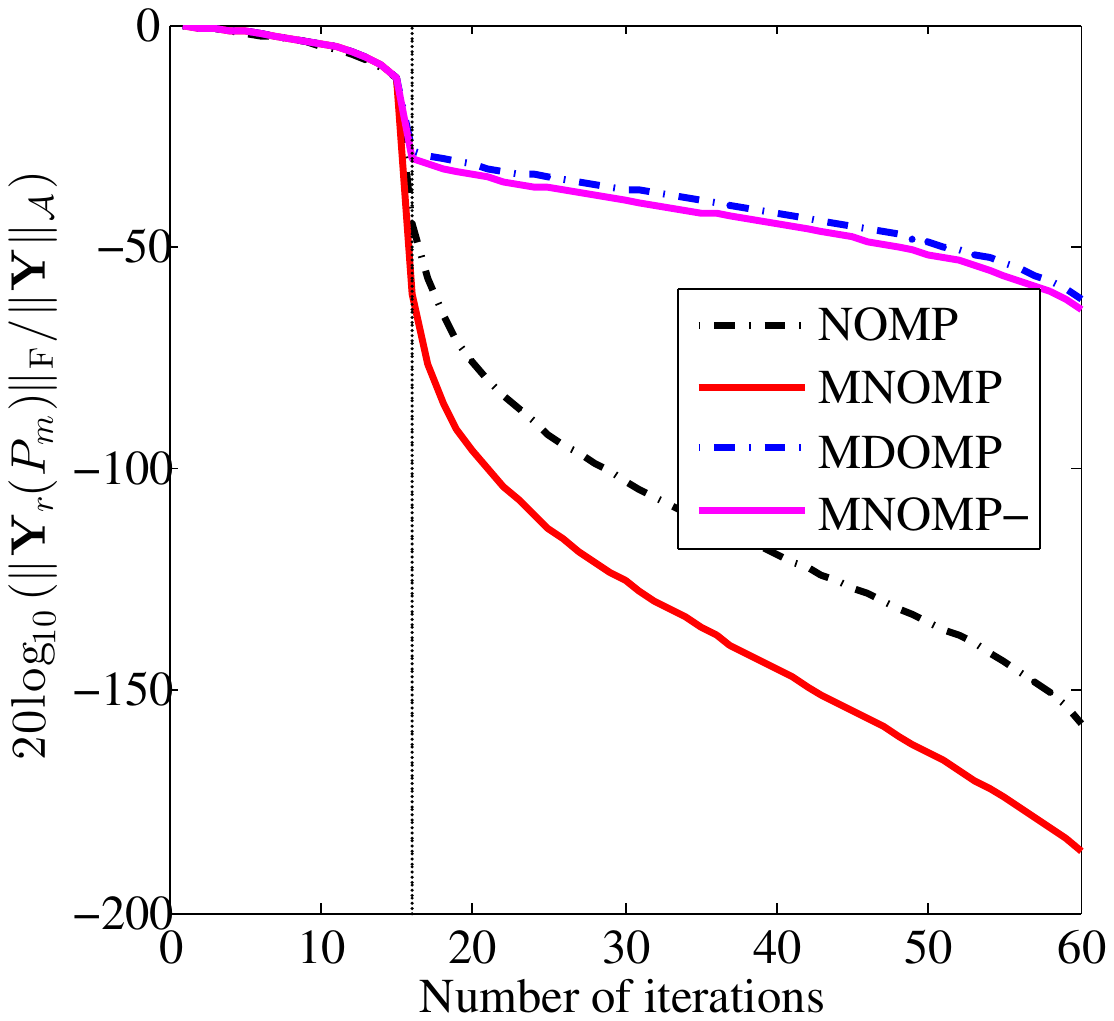}\\
  \caption{Convergence rates at noiseless case.}\label{Fig2}
\end{figure}
From Fig. \ref{Fig2}, it can be shown that a small oversampling factor $(\gamma=4)$ with single refinement step (MNOMP-) has a slightly faster convergence than a large oversampling factor $(\gamma=20)$ with no refinements (MDOMP). Furthermore, we can see a large gap between MNOMP and the other two algorithms, which means that the Cyclic Refinement Steps lead to an extremely fast convergence for MNOMP. In Fig. \ref{Fig2}, we can see that the relative residual energy drops sharply at the 16th iteration, which is equal to the number of sinusoids in the mixture $K$. Furthermore, increasing the number of snapshots improves the convergence rate.
\section{Simulation}
In this section, we conduct numerical simulations to compare the performance of MNOMP against other methods in terms of estimation accuracy in various scenarios.


\emph{Benchmarks:} All algorithms are compared against DFT method implemented by coarsely picking out the top $K$ peaks, and the oversampling rate for DFT is set to 4.


\emph{Newtonized BPDN:} For the sparse method inspired by sparse representation, we employ the SPGL1 toolbox \cite{Van2} to solve the $l_{2,1}$ minimization problem, also known as the MMV version of BPDN, we denote this method simply as BPDN in the following, which is defined as
\begin{align}
{\rm {min}}~\Vert\mathbf X \Vert_{1,2}~{\rm {subject~to}}~\Vert {\mathbf A}{\mathbf X} - {\mathbf Y} \Vert_{\rm F}\leq \tau,
\end{align}
where $\Vert\mathbf X\Vert_{1,2}$ denotes the sum of the two-norms of the rows of ${\mathbf X}$ and $\tau$ is the measure of the noise level. Newtonized BPDN can be viewed as an extension of the BPDN method. By applying this toolbox to model (\ref{multisnap}), where the oversampling rate is set to 4 and $\tau = \sqrt{NT\sigma^2}$, we then obtain the optimized $\mathbf X$. By sorting the $l_2$ norms of every row of the estimated $\mathbf X$ in a descending order, we choose the $K$ frequencies that correspond to the top $K$ $l_2$ norms of every row of the estimated $\mathbf X$ as the estimated frequencies. To avoid the frequency splitting phenomenon \cite{Yang2, Yang3}, we impose an extra procedure to the Newtonized BPDN method to cope with two special cases. On the one hand, when the interval between two adjacent estimated frequencies sorted in an ascending order is smaller or equal to $2\pi/(\gamma N)$, we eliminate the latter frequency. At the same time we add the frequency corresponding to the top ($K+1$)th $l_2$ norms of the row of the estimated $\mathbf X$. On the other hand, we also adopt the afore-mentioned step to improve the estimation performance when the first and last element of the sorted frequencies are equal to 0 and the last sampling frequency, respectively. We recycle the whole procedure until there is no occurrence of those two incidents. Then we apply the cyclic refinement step of MNOMP to the estimated frequencies, and we set $R_c = 1$, $R_s = 1$ in this case.

\emph{Atomic Norm Based Approaches:} The SPA \cite{Yang1}, reweighed atomic norm (RAM) \cite{Yang3}, the signal denoising for MMV model  \cite[equation (22)]{Chi} are implemented to make performance comparison. For clarity, the approach in \cite[equation (22)]{Chi} is termed as AST-SD. SPA and RAM are implemented by CVX \cite{CVX1}, while AST-SD is implemented via ADMM. For the three algorithms, we input $K$ to estimate the frequencies, while RAM automatically estimate the number of sinusoids.

\emph{MNOMP:} Parameters in various settings are summarized in Table \ref{table2}. The overestimating probability is set as $P_{\rm oe}=0.01$.

%
%
%

{\emph{{Simulation Set-up:}} The original frequencies of the sinusoids in the mixture are sampled from $[0, 2\pi)$ and meet the corresponding minimum separation criterion. The detailed settings are shown in Table \ref{table1}.}

{In the following, we aim at obtaining the frequency estimation accuracy performance of MNOMP as compared to the other methods in terms of mean squared error, recovery probability and DOA application.}
\begin{table}[h!t]
    \begin{center}
\caption{Settings of different scenarios.}\label{table1}\label{table2}
    \begin{tabular}{|c|c|c|c|c|}
            \hline
            Scenarios & SNR (${\rm {dB}}$) & $\Delta\omega_{\rm {min}}/\Delta\omega_{\rm {DFT}}$& $R_c$& $R_s$\\ \hline
            1 &   ${\rm SNR}_{\rm {nom}}$ & 2&1&1\\ \hline
            2 &   ${\rm SNR}_{\rm {nom}}$ & 1&3&1\\ \hline
             \end{tabular}
    \end{center}
\end{table}

\subsection{Normalized MSE vs number of snapshots T}
{ \emph{{Simulation Set-up:}} Here, we concentrate on estimating a mixture of $K=16$ complex sinusoids with length $N = 50$. Each scenario, as defined by $\Delta\omega_{\rm {min}}$ and SNR (\ref{snr}) is implemented over 300 MC trials. In scenario 1 and 2, ${\rm {SNR}_{\rm {nom}}}$ is set as 10 ${\rm {dB}}$ for all sinusoids in the mixture. It should be mentioned that all algorithms are compared against the DFT method, which is implemented by coarse peak picking. We regard the mean value of CRB of all $K$ sinusoids as a measure of optimality. The NMSE is defined as ${\rm E}\left[ \Vert {\boldsymbol \omega}_{\rm {est}} -{\boldsymbol \omega}_{\rm {true}}\Vert_2^2\right]/\Delta_{\rm {DFT}}^2$.}

Fig. \ref{SNAP1} and Fig. \ref{SNAP2} have demonstrated the NMSE of all algorithms versus the number of snapshots in scenario $1$ and scenario $2$, respectively. It is noted that for MNOMP and RAM algorithm, which assume that the model order $K$ is unknown, we only calculate the NMSE when the model order $K$ is successfully detected. Correspondingly, we present the recovery probability of both algorithms in Fig. \ref{SNAP1} (b) and Fig. \ref{SNAP2} (b).

For scenario $1$, it can be seen that all algorithms benefit from the snapshots. There exists a performance gap between the SPA and CRB, while MNOMP, RAM and NBPDN all asymptotically approach CRB. From Fig. \ref{SNAP1} (b), the recovery probability increases as the number of snapshots increases. It is noted that when $T$ is larger than $4$, the recovery probability of MNOMP is larger than that of RAM.

{When the signal is close to each other, we can basically sum up the same conclusion from Fig. \ref{SNAP2} compared to that from Fig. \ref{SNAP1}. The difference is that NBPDN arrives at CRB slower and there exists a small performance gap between all algorithms and CRB. }

\begin{figure}[htbp]\centering \subfigure[NMSE vs T.]{\begin{minipage}[t]{0.45\linewidth}\centering\includegraphics[width=2.5in]{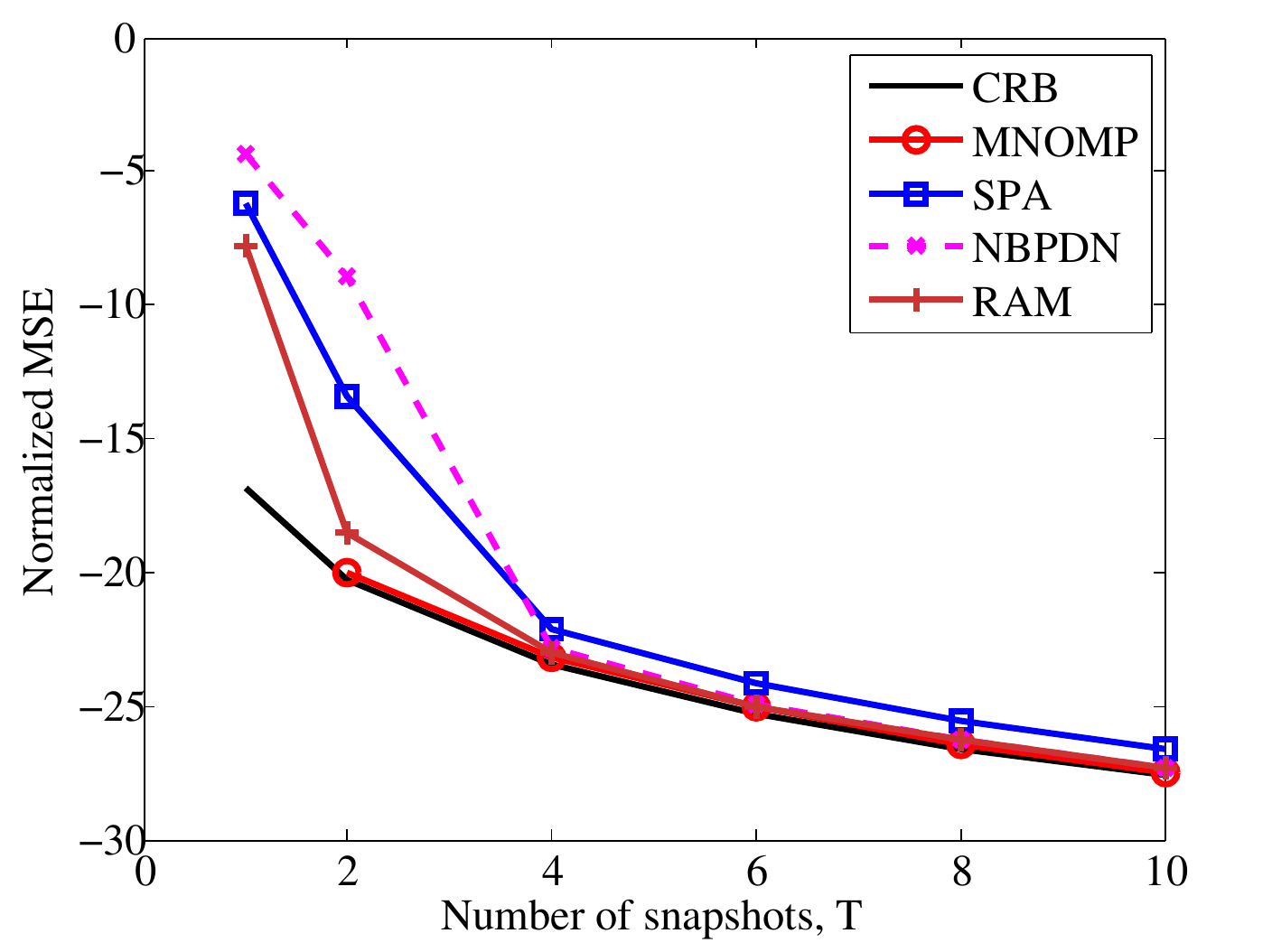}
\end{minipage}}
\subfigure[Recovery probability vs T.]{\begin{minipage}[t]{0.45\linewidth}\centering\includegraphics[width=2.5in]{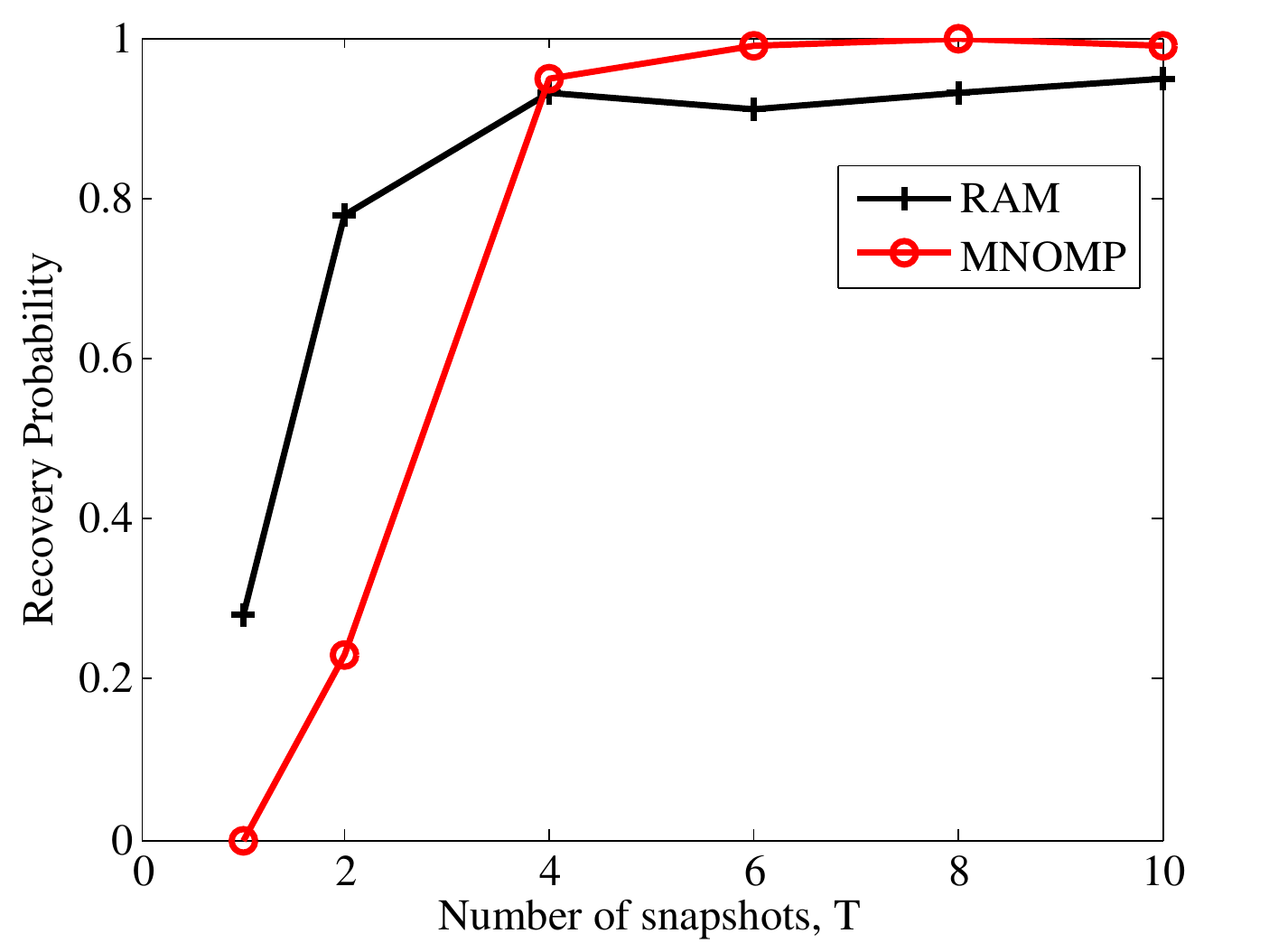}
\end{minipage}}              
\centering\caption{Normalized frequency MSE vs number of snapshots  for scenario 1.}\label{SNAP1}
\end{figure}

\begin{figure}[htbp]\centering \subfigure[NMSE vs T.]{\begin{minipage}[t]{0.45\linewidth}\centering\includegraphics[width=2.5in]{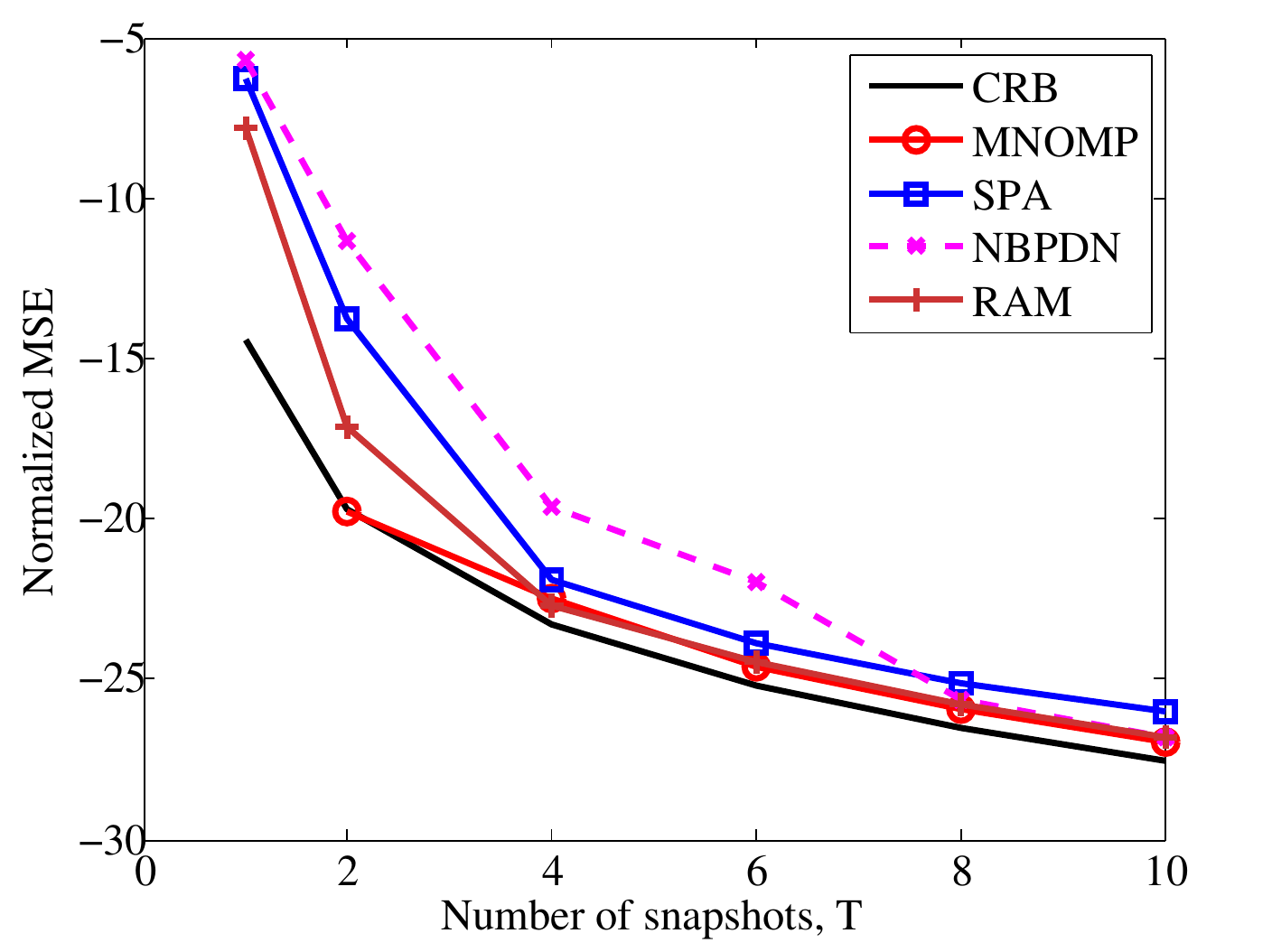}
\end{minipage}}
\subfigure[Recovery probability vs T.]{\begin{minipage}[t]{0.45\linewidth}\centering\includegraphics[width=2.5in]{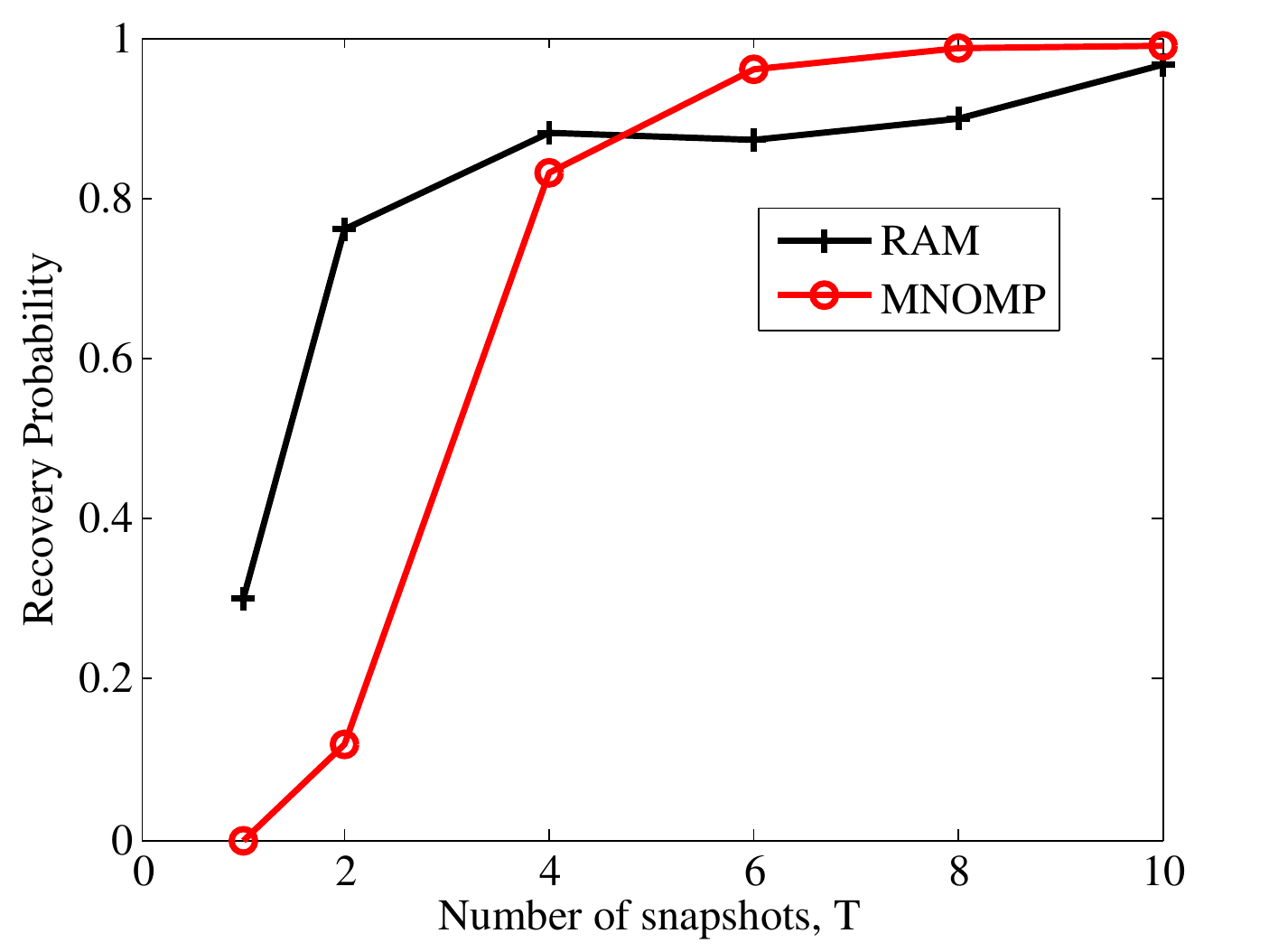}
\end{minipage}}              
\centering\caption{Normalized frequency MSE vs number of snapshots  for scenario 2.}\label{SNAP2}
\end{figure}

\subsection{Application: DOA Estimation}
{In this section, we investigate the estimation performance of MNOMP algorithm in the DOA scenario, where the inter-element spacing $d$ is half of the wavelength $\lambda$, i.e., $d = \lambda/2$. We consider $K=3$ narrow band far-field DOA angles ${\boldsymbol \phi}=[-2^{\circ}, 5^{\circ}, 12^{\circ}]^{\rm T}$. We set $N = 40$, $T = 20$. For small sample scenario, the EPUMA outperforms many other subspace based DOA estimators and offers reliable performance with small number of samples \cite{Qian}. Thus, we compare MNOMP with EPUMA. In this part, the root MSE (RMSE) $ \sqrt{\sum_{k = 1}^K (\hat{\phi}_k - \phi_k)^2}$ is used to characterize the performance of MNOMP and EPUMA, where $\hat{\phi}$ denotes the output of the algorithm. The results are presented in Fig. \ref{RMSE} (a). It can be seen that MNOMP performs better than that of EPUMA. All these algorithms approach CRB as SNR increases. From Fig. \ref{RMSE} (b), we can conclude that as the SNR increases, the recovery probability approaches one as SNR increases.}
\begin{figure}[htbp]\centering \subfigure[RMSE vs SNR.]{\begin{minipage}[t]{0.45\linewidth}\centering\includegraphics[width=2.7in]{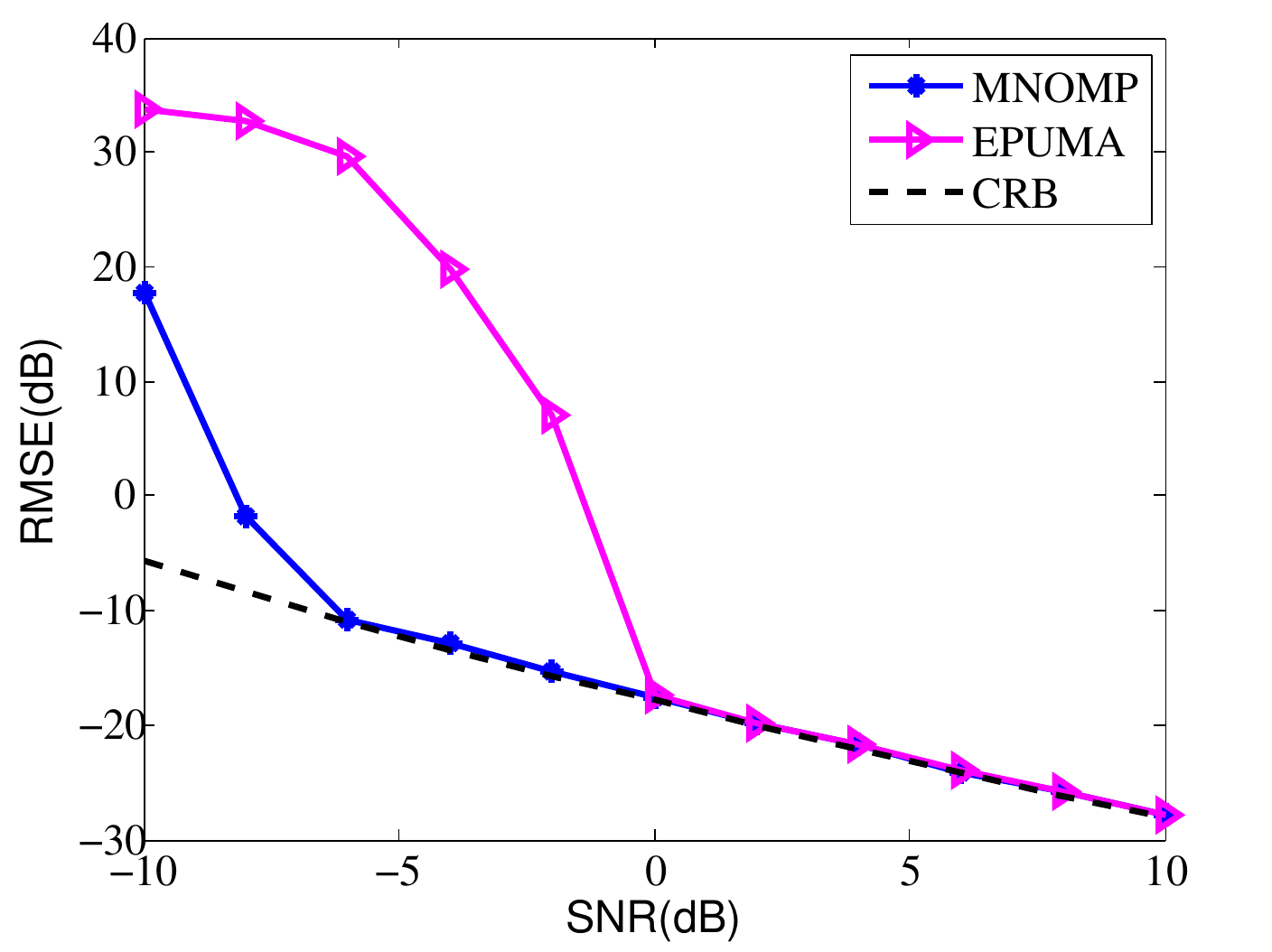}
\end{minipage}}
\subfigure[Recovery probability of MNOMP.]{\begin{minipage}[t]{0.45\linewidth}\centering\includegraphics[width=2.4in]{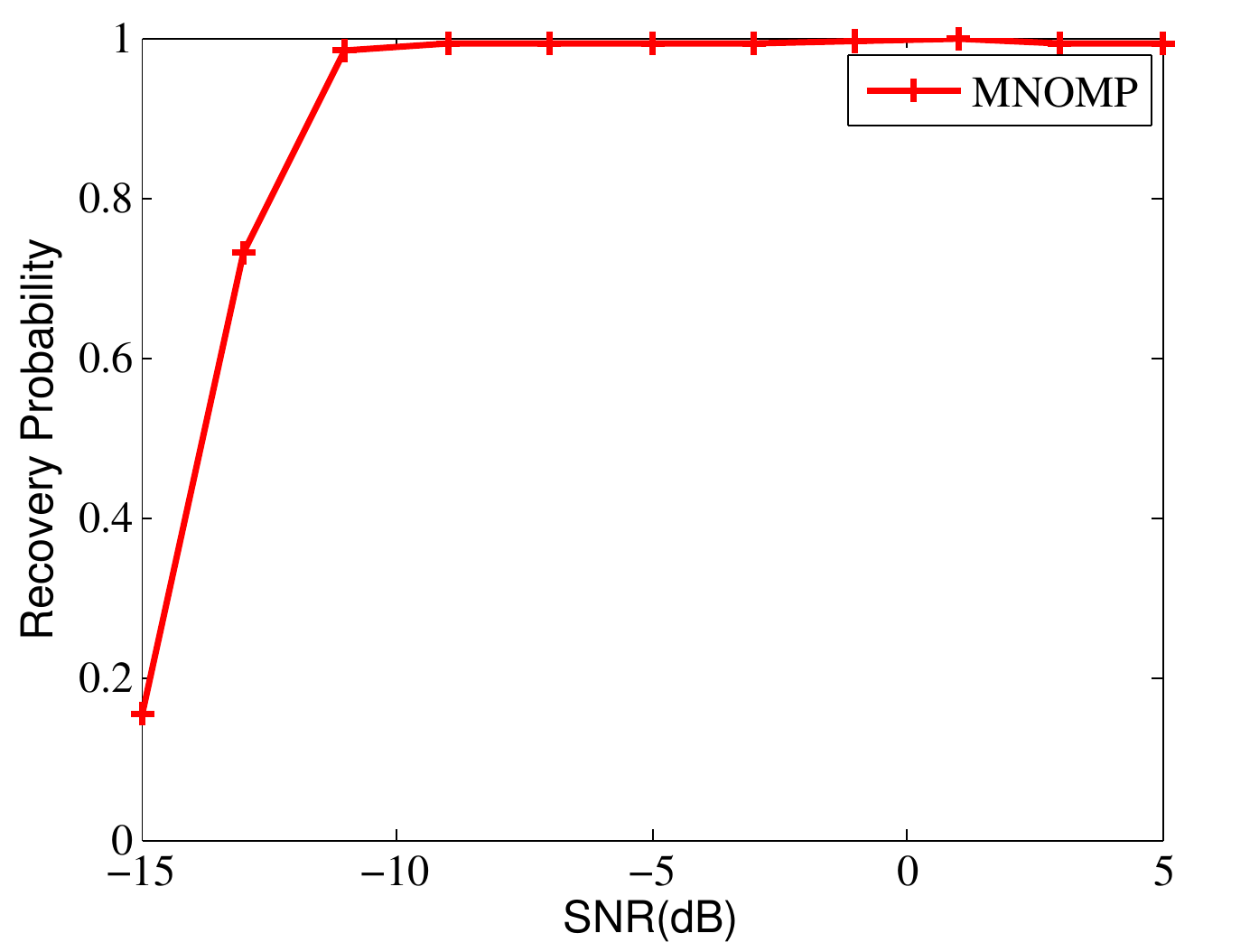}
\end{minipage}}              
\centering\caption{RMSE of MNOMP algorithm for DOA estimation.}\label{RMSE}
\end{figure}
\subsection{Success rate of all algorithms for multi-snapshot frequency estimation}
For the MMV model, the frequencies $\{\omega_k\}$ and amplitudes $\{x_{kt}\}$ consist of $(2T + 1)K$ real unknowns, and the minimum number of complex-valued samples is $\frac{1}{2}(2T+1)K$. Institutively, to recover these unknowns, the sample size per snapshot for any method must satisfy \cite{Atom3}
\begin{align}\label{lower}
N\geq \frac{1}{2T}(2T+1)K=K\left(1+\frac{1}{2T}\right).
\end{align}
Thus it is interesting to investigate the benefits of $N$ and $T$ for the various algorithms. In our simulations, the number of frequencies is $K = 10$, and the frequencies $\{\theta_k\}_{k=1}^K$ are randomly generated and satisfy $\Delta\omega_{\rm {min}}/\Delta\omega_{\rm {DFT}} = 1.2$. The frequencies are said to be successfully estimated if the model order is successfully estimated and the root MSE \cite{Atom3}, computed as $\sqrt{{\sum_{k = 1}^K|\theta_k - \hat{\theta}_k|^2}/{K}}
$ is less than $10^{-3}$, where $\hat{\theta}_k$ denotes the estimate of $\theta_k$. Results are presented in Fig. \ref{success}. It can be seen that the required sample size per snapshot for exact frequency estimation decreases as the number of snapshots increases. From Fig. \ref{success2}, we can see that RAM has the best recovery performance. MNOMP and SPA have the similar success rate performance. While the recovery performance of AST-SD is the worst.

\begin{figure}
 \centering
  \includegraphics[width=80mm]{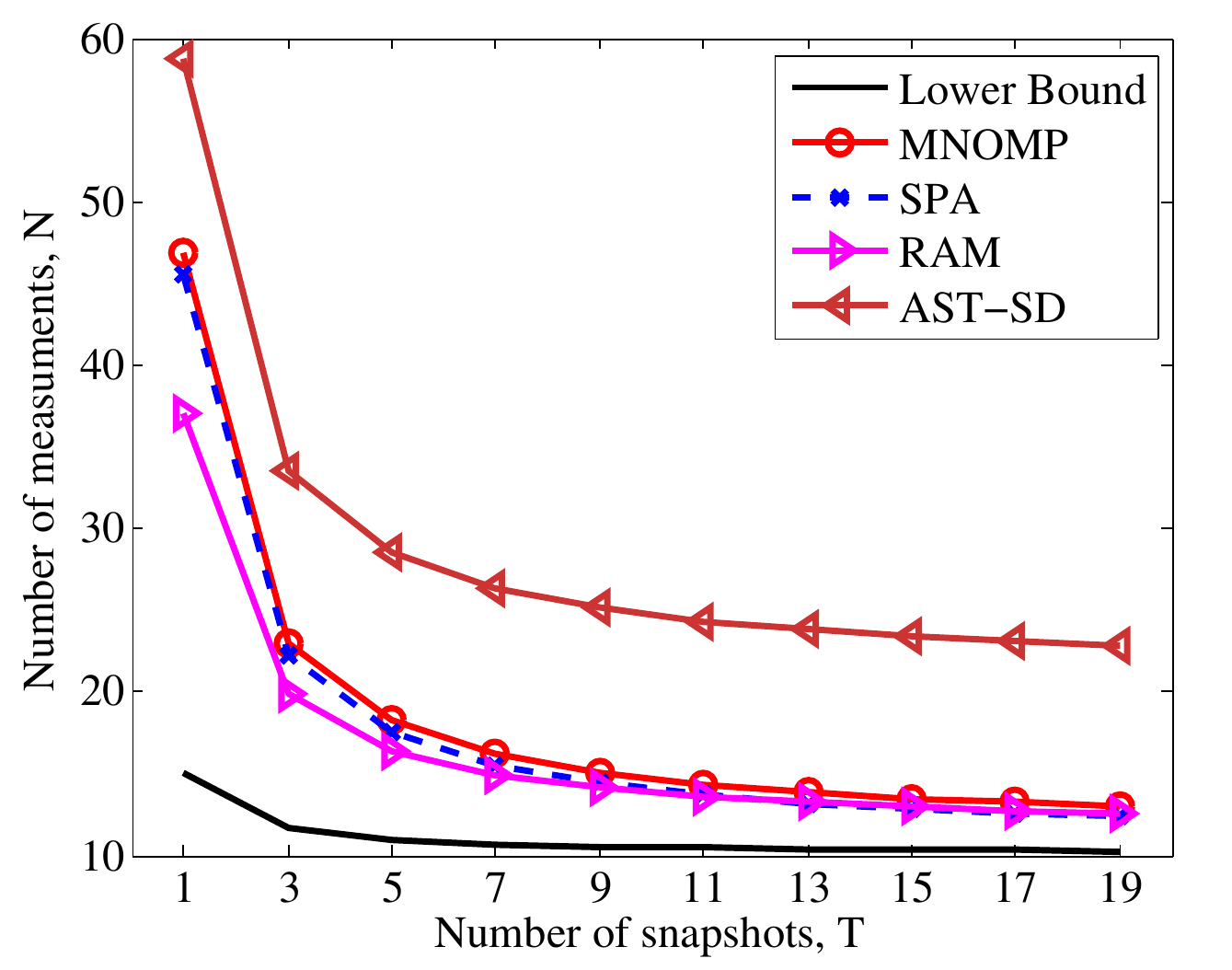}\\
  \caption{Success rates curves of all algorithms for line spectrum estimation with MMVs.}\label{success}
\end{figure}

\begin{figure}
 \centering
  \includegraphics[width=80mm]{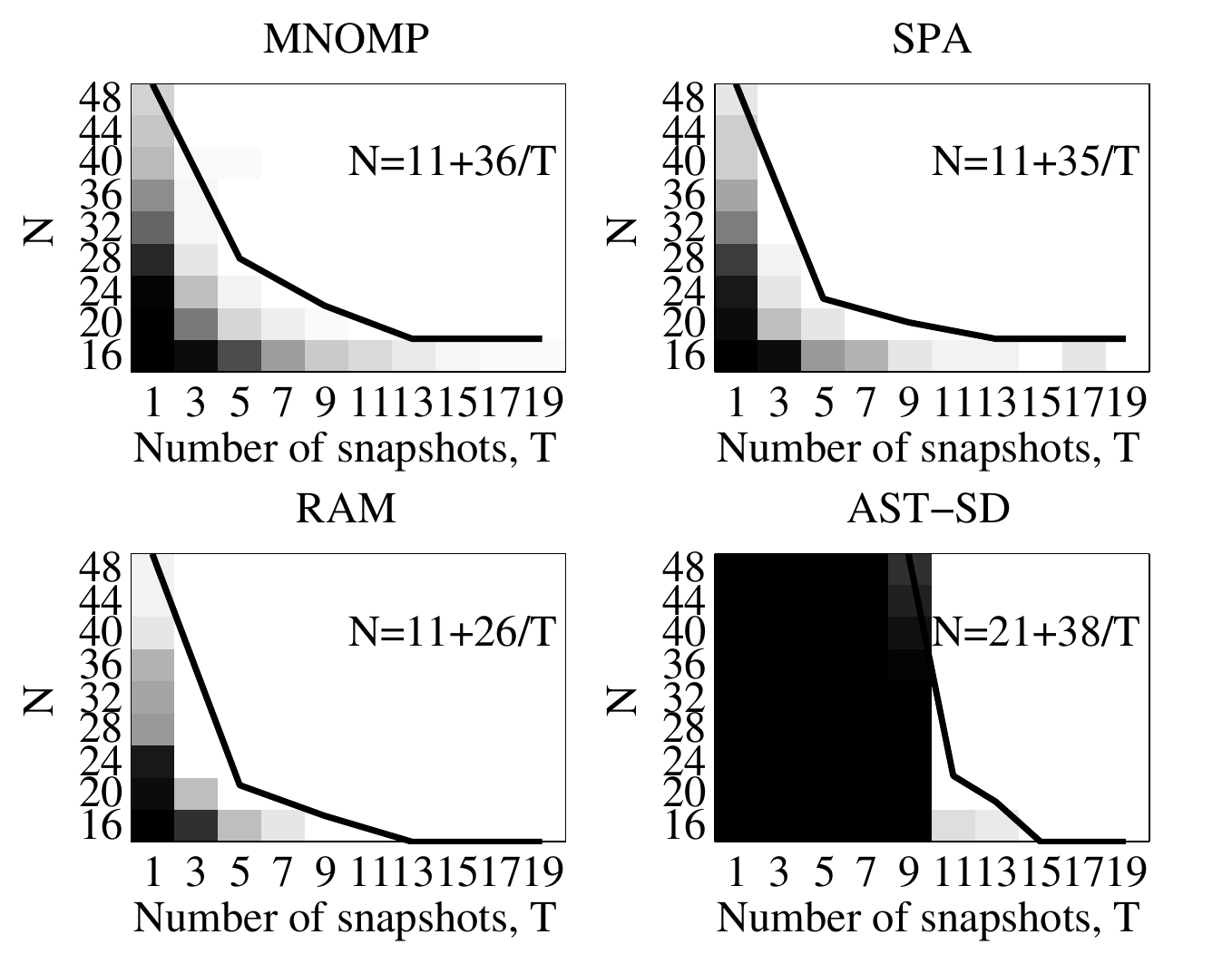}\\
  \caption{Success rates of all algorithms for line spectrum estimation with MMVs. White means complete success and black means complete failure.}\label{success2}
\end{figure}

\section{Conclusion}
This paper develops the MNOMP algorithm to deal with MMVs and show the benefit from MMVs numerically. The algorithm uses the information of already detected frequencies to refine the current frequency and has a stopping criterion based on a given overestimating probability. The convergence results are provided. It is numerically shown that MNOMP is competitive in terms of frequency estimation accuracy and asymptotically approaches CRB. Since the number of snapshots is very large in array signal processing, referring to the dimension reduction method \cite{Yangzaibook} to implement MNOMP is worth studying and will be left for future work.
\section{APPENDIX: Estimation theoretic bounds}
The complex CRB can be calculated by following the procedure similar to \cite{han} and utilize \cite[Example 2]{Nagesha11}. Here we provide an alternative approach to compute CRB. By defining $x_{kt} = g_{kt}e^{j\phi_{kt}}, \forall k = 1, \cdots, K, t = 1, \cdots, T$, We obtain matrices $\mathbf G$ and $\boldsymbol \Phi$. Let $\boldsymbol \kappa$ be ${\boldsymbol \kappa}=\left[\boldsymbol \omega^{\rm T}, {\mathbf g}^{\rm T}, \boldsymbol{\phi}^{\rm T} \right]^{\rm T}$, where ${\mathbf g}={\rm vec}({\mathbf G})$ and $\boldsymbol{\phi}={\rm vec}({\boldsymbol \Phi})$. Then the FIM is calculated according to \cite{Chi11}
\begin{align}\label{unqFIM}
{\mathbf I}({\boldsymbol \kappa})=\frac{2}{\sigma^2}\sum\limits_{n=1}^N \sum\limits_{t=1}^T  \left(\frac{\partial \Re\{Z_{nt}\}}{\partial {\boldsymbol \kappa}}\left(\frac{\partial \Re\{Z_{nt}\}}{\partial {\boldsymbol \kappa}}\right)^{\rm T}+\frac{\partial \Im\{Z_{nt}\}}{\partial {\boldsymbol \kappa}}\left(\frac{\partial \Im\{Z_{nt}\}   }{\partial {\boldsymbol \kappa}}\right)^{\rm T}\right).
\end{align}
By defining ${\mathbf g}_{t}=[g_{1t},\cdots,g_{Kt}]^{\rm T}$ and $\phi_{t}=[\phi_{1t},\cdots,\phi_{Kt}]^{\rm T}$, we have
\begin{align}\label{vecgrad}
\frac{\partial \Re\{Z_{nt}\}}{\partial {\boldsymbol \kappa}}=\left[\begin{array}{c}
                                                                         \frac{\partial \Re\{Z_{nt}\}}{\partial {\boldsymbol \theta}} \\
                                                                         {\mathbf 0}_{(t-1)K} \\
                                                                         \frac{\partial \Re\{Z_{nt}\}}{\partial {\mathbf g}_{t}} \\
                                                                         {\mathbf 0}_{(T-t)K} \\
                                                                         {\mathbf 0}_{(t-1)K} \\
                                                                         \frac{\partial \Re\{Z_{nt}\}}{\partial {\boldsymbol \phi}_{t}} \\
                                                                         {\mathbf 0}_{(T-t)K} \\
                                                                       \end{array}\right],~\frac{\partial \Im\{Z_{nt}\}}{\partial {\boldsymbol \kappa}}=\left[\begin{array}{c}
                                                                         \frac{\partial \Im\{Z_{nt}\}}{\partial {\boldsymbol \theta}} \\
                                                                         {\mathbf 0}_{(t-1)K} \\
                                                                         \frac{\partial \Im\{Z_{nt}\}}{\partial {\mathbf g}_{t}} \\
                                                                         {\mathbf 0}_{(T-t)K} \\
                                                                         {\mathbf 0}_{(t-1)K} \\
                                                                         \frac{\partial \Im\{Z_{nt}\}}{\partial {\boldsymbol \phi}_{t}} \\
                                                                         {\mathbf 0}_{(T-t)K} \\
                                                                       \end{array}\right],
\end{align}
where
\begin{subequations}
\begin{align}
&\frac{\partial \Re\{Z_{nt}\}}{\partial {\theta_k}} = -(n-1)g_{kt}{\rm {sin}}\left[(n-1)\theta_k + \phi_{kt} \right],\notag\\
&\frac{\partial \Re\{Z_{nt}\}}{\partial {g_{kt}}} = {\rm {cos}}\left[(n-1)\theta_k + \phi_{kt} \right],\notag\\
&\frac{\partial \Re\{Z_{nt}\}}{\partial {\phi_{kt}}} = -g_{kt}{\rm {sin}}\left[(n-1)\theta_k + \phi_{kt} \right],\notag\\
&\frac{\partial \Im\{Z_{nt}\}}{\partial {\theta_k}} = (n-1)g_{kt}{\rm {cos}}\left[(n-1)\theta_k + \phi_{kt} \right],\notag\\
&\frac{\partial \Im\{Z_{nt}\}}{\partial {g_{kt}}} = {\rm {sin}}\left[(n-1)\theta_k + \phi_{kt} \right],\notag\\
&\frac{\partial \Im\{Z_{nt}\}}{\partial {\phi_{k,t}}} = g_{kt}{\rm {cos}}\left[(n-1)\theta_k + \phi_{kt} \right].\notag
\end{align}
\end{subequations}
Substituting (\ref{vecgrad}) in (\ref{unqFIM}), the FIM ${\mathbf I}({\boldsymbol \kappa})$ is obtained. The CRB is ${\rm CRB}({\boldsymbol \kappa})={\mathbf I}^{-1}({\boldsymbol \kappa})$ and CRB of frequencies are $[{\rm CRB}({\boldsymbol \kappa})]_{1:K,1:K}$, which will be used as the performance metrics.
\section{Acknowledgement}
This work is supported by Zhejiang Provincial Natural Science Foundation of China under Grant LQ18F010001.

\end{document}